\documentclass[11pt]{article}
\usepackage[top=1in,bottom=1in,left=1in,right=1in]{geometry} 
\usepackage{palatino}

\usepackage[T1]{fontenc}
\usepackage{lmodern}
\usepackage[american]{babel}
\usepackage{tabularx}
\usepackage{booktabs}
\usepackage[table]{xcolor}
\usepackage{algorithm}
\usepackage{algpseudocode}

\usepackage{enumerate}
\usepackage{verbatim}
\usepackage{natbib}
\usepackage{amsmath}
\usepackage{subfigure}
\usepackage{enumerate}
\usepackage{subfigure}
\usepackage{verbatim}
\usepackage{natbib}
\usepackage{multirow}
\usepackage{graphicx}
\usepackage{amssymb,amsmath}
\usepackage{color}
\usepackage{hyperref}
\usepackage{url}

\normalsize

\def\be{\begin{eqnarray}}
\def\ee{\end{eqnarray}}
\usepackage{epsfig}

\newtheorem{lem}{Lemma}

\title{Group Sequential Sample Size for Comparing Two Survival Probabilities at a Specific Time Point}

\author{Susan Halabi$^\ast$, Lu Liu, Chenxi Yu, and Yuan Wu}

\begin{document}

\maketitle

\begin{center} 
Department of Biostatistics and Bioinformatics, Duke University,\\ Durham, North Carolina, USA\\[2pt] $^\ast$\texttt{Corresponding author: susan.halabi@duke.edu}
\end{center}

\begin{abstract}
We propose a novel method that simultaneously determines the sample size for testing two survival probabilities at a pre-specified ltime  while guaranteeing type I error control in both fixed and group-sequential trial designs. Simulations across varying hypothesized differences, failure distributions, censoring proportions, and nominal powers demonstrate consistent performance, while interim analyses highlight reduced type I error and increased power at each look, regardless of the underlying failure time distribution or spending function. Importantly, our method is especially useful for evaluating survival outcomes at a fixed time in randomized trials where one treatment arm includes neoadjuvant therapy prior to surgery while the other involves surgery alone. Furthermore, it is advantageous when the proportional hazards assumption is not satisfied, as often occurs in immunotherapy trials with delayed or time-varying treatment effects or crossing survival curves. The method is also applicable to randomized phase II trials, where smaller sample sizes and the use of intermediate or surrogate time-to-event endpoints demand efficient data use and robust error control. We illustrate the approach with motivating examples in renal and prostate cancer. An accompanying R Shiny application enables investigators to compute sample sizes interactively, facilitating practical trial planning in diverse settings.

\end{abstract}

\section{Introduction}
The primary outcome of a clinical trial evaluating a therapeutic, diagnostic, or other medical intervention is central to defining the intervention’s clinical benefit. The choice of the primary endpoint influences key aspects of trial design, including sample size, trial duration, frequency of assessments, and cost. Demonstration of efficacy typically requires a tangible clinical benefit, such as prolonged survival, symptom improvement, or pain reduction. Phase III clinical trials often use overall survival (OS) as the gold-standard endpoint, although other clinically meaningful time-to-event outcomes, such as progression-free or event-free survival, are also commonly used. \\

In many clinical trials, endpoints are assessed at a prespecified landmark time rather than over the entire follow-up period \cite{{eastham:2020},{pfister:2022}, {kenter:2023}}. This approach is particularly relevant in adjuvant or neoadjuvant trials, where patients are treated with curative intent, typically via surgery with or without systemic therapy, and OS may take many years to mature. Early after surgery, death events are uncommon, and patients may survive for extended periods even after disease recurrence. In contrast, disease-free survival (DFS) or PFS) evaluated at a fixed time can be observed much sooner (e.g., 3 years after randomization), enabling faster trial readout and more timely decision-making for regulatory approval and clinical adoption. For example, the CALGB 90203 (PUNCH) phase III trial used 3-year biochemical PFS as the primary endpoint to assess whether six months of neoadjuvant docetaxel plus androgen deprivation therapy (ADT) prior to radical prostatectomy improved outcomes compared with immediate surgery \cite{eastham:2020}. Because ADT directly affects prostate-specific antigen (PSA) levels, a milestone time to event endpoint was chosen to reduce potential bias from differential treatment effects between the study arms. Similarly, in a recent bladder cancer trial, the primary endpoint was 3-year PFS \cite{pfister:2022}. In that study, 500 patients with non-metastatic muscle-invasive bladder cancer were randomized to receive either methotrexate, vinblastine, doxorubicin, and cisplatin every two weeks or four cycles of gemcitabine and cisplatin every three weeks in the neoadjuvant setting, or surgery alone, illustrating how a fixed-time endpoint can meaningfully capture differences between treatment strategies.

Against this backdrop, we seek to design a phase III trial in patients with a small renal mass suspected of renal cancer. Patients will be randomized to undergo diagnostic biopsy or no biopsy, in addition to standard radiographic work-up. The primary endpoint is 5-year intervention-free survival (IFS), a composite time-to-event outcome evaluated at a fixed time point. In this article, we propose a novel statistical framework for designing trials with time-to-event endpoints measured at a specific time. Building on the group sequential design for single-arm trials proposed by Lin et al. (1996)~\cite{Lindanyu:1996}, we extend the methodology to two-arm randomized trials and demonstrate its application to both fixed and sequential designs. Our approach enables more accurate variance estimation, supports interim monitoring, and improves trial efficiency without compromising statistical validity. The remainder of the article is organized as follows: Section 2 introduces the proposed methodology; Section 3 presents the simulation study; Section 4 provides real-world examples; and Section 5 discusses the implications of our findings. \\

\section{Methods}
\subsection{Background}\label{sec2.1}
 Let $Y_i$, $T_i$, and $C_i$ be the mutually independent enrollment time, event time, and censoring time of individual $i$, respectively, for $i = 1,\ldots,n$, where the measurements of both $T_i$ and $C_i$ start at $Y_i$. Assume that the collection $\{(Y_i, T_i, C_i)\}_{i=1}^n$ consists of independent and identically distributed triples.

At calendar time $t$, when an interim analysis is planned, the observed time for subject $i$ is 
\[
X_i(t) = \min\{T_i,\, C_i,\,(t - Y_i)^{+}\},
\]
and the failure indicator is
\[
\Delta_i(t) = I\!\left[T_i \le \min\{C_i,\,(t - Y_i)^{+}\}\right].
\]
Using the dataset $\{X_i(t), \Delta_i(t)\}_{i=1}^n$, the Nelson--Aalen estimator of the cumulative hazard function 
$\Lambda(x) = \int_{0}^{x}\lambda(u)\,du$ is
\begin{equation}\label{Lambda_hat}
\hat{\Lambda}(x;t)
    = \sum_{i : X_i(t) \le x} \frac{\Delta_i(t)}{R_i(t)},
\end{equation}
where $R_i(t) = \sum_{j=1}^{n} I\!\left[X_j(t) \ge X_i(t)\right]$.

As described by Lin  et al. (1996)~\cite{Lindanyu:1996}, for a fixed calendar time $t$ and a failure time $x \le t$, 
the process $n^{1/2}\{\hat{\Lambda}(x;t) - \Lambda(x)\}$ converges weakly to a zero-mean Gaussian 
process with independent increments and variance function
\begin{equation}\label{sigma}
\sigma^2(x;t) = \int_{0}^{x} \frac{\lambda(u)\,du}{\pi(u;t)},
\end{equation}
where $\pi(u;t) = P\{X_i(t) \ge u\}$. An estimator of $\sigma^2(x;t)$ is given by
\begin{equation}\label{sigmahat}
\hat{\sigma}^2(x;t)
    = \sum_{i : X_i(t) \le x} \frac{\Delta_i(t)}{R_i^2(t)/n}.
\end{equation}

\subsection{Test Statistic}

Suppose that the study consists of two treatment arms: arm~$1$ and $2$ with $n_1$ and $n_2$ individuals, respectively, where $n_1 + n_2 = n$.
From the weak convergence described in Section~\ref{sec2.1}, regarding cumulative hazard functions $\Lambda_1(\cdot)$ and $\Lambda_2(\cdot)$ in arm 1 and 2, respectively, we have
 $$\left\{\hat{\Lambda}_2(x;t)-\Lambda_2(x)\right\}-\left\{\hat{\Lambda}_1(x;t)-\Lambda_1(x)\right\}\approx
N\left\{0,\frac{\sigma_1^2(x;t)}{n_1}+\frac{\sigma_2^2(x;t)}{n_2}\right\},$$  where $\hat{\Lambda}_i(x;t)$ as defined by (\ref{Lambda_hat}),  is the estimator for $\Lambda_i(x)$ at  monitoring time $t$ in arm $i$ for $i=1, 2$,  
$\sigma_i^2(x;t)$ is as 
displayed by (\ref{sigma}) for arm $i$ with $i=1, 2$.  
Then, the standardized test statistic of the difference between the two arms is
	\begin{equation*}
	Z(x;t)=\frac{\hat{\Lambda}_2(x;t)-\hat{\Lambda}_1(x;t)}{\sqrt{\frac{\hat{\sigma}_1^2(x;t)}{n_1}+\frac{\hat{\sigma}_2^2(x;t)}{n_2}}},
	\end{equation*}
where $\hat{\sigma}_i^2(x;t)$ is as presented by (\ref{sigmahat}) for arm $i$ with $i=1, 2$.

Alternatively, by the delta method, $\left[\log\left\{\hat{\Lambda}_2(x;t)\right\}-\log\left\{\Lambda_2(x)\right\}\right]-\left[\log\left\{\hat{\Lambda}_1(x;t)\right\}-\log\left\{\Lambda_1(x)\right\}\right]$ can be approximated by $N\left(0,\Sigma_{\mathrm{t}}\right)$, where  
\begin{equation}\label{variance_log}
\Sigma_{\mathrm{t}}=\frac{\sigma_1^2(x;t)}{n_1\left[\log\left\{S_1(x)\right\}\right]^2}+\frac{\sigma_2^2(x;t)}{n_2\left[\log\left\{S_2(x)\right\}\right]^2},
\end{equation}
with $S_1(\cdot)=\exp\left\{-\Lambda_1(\cdot)\right\}$ and $S_2(\cdot)=\exp\left\{-\Lambda_2(\cdot)\right\}$. Consequently, the following test statistic can  be derived. 
	\begin{equation}\label{statistic}
	Z(x;t)=\frac{\log\left\{\hat{\Lambda}_2(x;t)\right\}-\log\left\{\hat{\Lambda}_1(x;t)\right\}}{\sqrt{\frac{\hat{\sigma}_1^2(x;t)}{n_1\hat{\Lambda}_1^2(x;t)}+\frac{\hat{\sigma}_2^2(x;t)}{n_2\hat{\Lambda}_2^2(x;t)}}}
	\end{equation}
Note that for the rest of this manuscript we choose the the test statistic $Z(x;t)$ defined by (\ref{statistic}) since the logarithm  of cumulative hazard function is not restricted to positive values.

\subsection{Group Sequential Test}\label{sec2.3}

For a $K$-stage group sequential design, let $p_k$ (for $k=1,\dots,K$) denote the percentage of individuals with at least follow-up time $x$ at the $k$th interim  look. For ease of presentation, we refer to the set $\{p_k\}_{k=1}^K$ as ``complete percentages''.
 It is evident that the monitoring time $t_k=\min\left\{(np_k)/r+x, n/r+t_f\right\}$, where $n$ is the total sample size of enrolled patients, $r$ is the constant accrual rate, $t_f$ is the follow-up duration after the accrual. 

With the sequence of test statistics  $\left\{Z\left(x;t_k\right)\right\}_{k=1}^K$ based on (\ref{statistic}), the proposed sequential test is determined by the corresponding critical boundaries $\{c_k\}_{k=1}^K$. In what follows, we describe how to derive $\{c_k\}_{k=1}^K$. Without loss of generality, we consider two-sided tests for the remainder of the manuscript.

First, the O'Brien–Fleming (1979)~\cite{obrien:fleming:1979} or Lan-DeMets (1983)~\cite{lan:demets:1983} approach can be employed for the alpha spending among $K$ tests to control the type I error while maintaining an increasing and convex spending trend. Specifically, in the O'Brien-Fleming and Lan-DeMets approaches, the cumulative type I error function are defined by 
$\alpha\left(t_k\right)=2-2\cdot\Phi\left\{\Phi^{-1}\left(1-\alpha/2\right)\big/\sqrt{p_k}\right\}$ and $\alpha\left(t_k\right)=\left(\Sigma_{t_K}\big/\Sigma_{t_k}\right)^2\cdot\alpha$, respectively, where $\alpha$ is the overall nominal type I error level, $\Phi(\cdot)$ is the cumulative probability function of the standard normal distribution, and $\Sigma_{t_k}$ is as defined by (\ref{variance_log}). In both approaches, the alpha level at $t_k$ is calculated by $$\alpha_{1}=\alpha\left(t_1\right)~\mathrm{or}~ \alpha_{k}=\alpha(t_k)-\alpha\left(t_{k-1}\right)~\mathrm{for}~k=2,\cdots,K.$$

Next, 
under the null hypothesis $H_0: S_1(x)=S_2(x)=s_0$, for constant $s_0$ satisfying $0<s_0<1$, parallel to the result for the one-sample group sequential design in Lin et al. (1996)~\cite{Lindanyu:1996}, we have the following approximation of
the vector of the proposed test statistics as defined by (\ref{statistic}).
\begin{align*}
    \left\{Z\left(x;t_1\right),Z\left(x;t_2\right),\cdots,Z\left(x;t_K\right)\right\}^\top \approx \left(G_1,G_2,\cdots,G_K\right)^\top\sim N\left(\boldsymbol{0}_{K\times 1}, \boldsymbol{\Sigma}_{K\times K}\right),
\end{align*}
where 
\begin{equation}\label{covariancematrix}
    \boldsymbol{\Sigma}_{K\times K}=\left(\begin{array}{cccc}
1&\sigma_{12}&\cdots&\sigma_{1K}\\
\sigma_{21}&1&\cdots&\sigma_{2K}\\
\vdots&\vdots&\ddots&\vdots\\
\sigma_{K1}&\sigma_{K2}&\cdots&1
\end{array}\right),
\end{equation}
with $\sigma_{uv}=\sigma_{vu}=\sqrt{\Sigma_{t_v}\big/{\Sigma_{t_u}}}$ ($1\le u<v\le K$), for both $\Sigma_{t_u}$ and $\Sigma_{t_v}$ as displayed by (\ref{variance_log}).

Consequently,  $\{c_k\}_{k=1}^K$  can be obtained by solving the following equation   iteratively  with $k$ starting at $1$, based on the multivariate normal distribution of $\left(G_1,G_2,\cdots,G_K\right)^\top$.
\begin{equation}\label{criticalboundaries}
	\Pr\left(\left|G_1\right|\le c_1,\cdots,\left|G_{k-1}\right|\le c_{k-1},\left|G_k\right|>c_k\right)=\alpha_k.
	\end{equation}
    
    


\subsection{Power and Sample Size Calculation}

The power calculation for the proposed sequential design relies on the asymptotic behavior of the test statistic $\left\{Z\left(x;t_k\right)\right\}_{k=1}^K$ under $S_1(x)\neq S_2(x)$. Specifically, the following  approximation (\ref{approx}) facilitates the power calculation. Similar to the approximation under the null hypothesis $H_0$ described in Section 2.3 , under the alternative hypothesis $H_1: S_1(x)=s_0, S_2(x)=s_1$, for constant $s_0$ and $s_1$ satisfying $0<s_0<1$, $0<s_1<1$, $s_0\neq s_1$,  we have 
\begin{align}\label{approx}
    \left\{Z\left(x;t_1\right),Z\left(x;t_2\right),\cdots,Z\left(x;t_K\right)\right\}^\top \approx \left(\tilde{G}_1,\tilde{G}_2,\cdots,\tilde{G}_K\right)^\top\sim N\left(\boldsymbol{\mu}_{K\times 1}, \boldsymbol{\Sigma}_{K\times K}\right),
\end{align}
where 
$\boldsymbol{\Sigma}_{K\times K}$ is as defined by (\ref{covariancematrix}),
\begin{equation}\label{mean}
\boldsymbol{\mu}^\top_{K\times 1}=\left(\frac{h(s_1)-h(s_0)}{\sqrt{\Sigma_{t_1}}}, \frac{h(s_1)-h(s_0)}{\sqrt{\Sigma_{t_2}}},\cdots,\frac{h(s_1)-h(s_0)}{\sqrt{\Sigma_{t_K}}}\right), \end{equation}
for $h(\cdot)=\log\left\{-\log\left(\cdot\right)\right\}$ and  $\Sigma_{t_k}$ as displayed by (\ref{variance_log}) with $k=1,\cdots,K$.

Based on the multivariate normal distribution of $\left(\tilde{G}_1,\tilde{G}_2,\cdots,\tilde{G}_K\right)^\top$, we use the critical boundaries $\left\{c_k\right\}_{k=1}^K$ from the previous section to evaluate the reject probabilities $\left\{\mathrm{RP}_k\right\}_{k=1}^K$ through $K$ monitoring looks as follows.
\begin{equation}\label{power}
\mathrm{RP}_k=\Pr\left(\left|\tilde{G}_1\right|\le c_1,\cdots,\left|\tilde{G}_{k-1}\right|\le c_{k-1},\left|\tilde{G}_k\right|>c_k\right).
\end{equation}
The theoretical power denoted by $p$ is then calculated as  the sum of these rejection probabilities, that is, $p=\sum_{k=1}^K\mathrm{RP}_k$.

Because a sequential design does not admit the straightforward inversion of the power calculation used in fixed designs, we propose Algorithm~\ref{alg1} for sample-size determination. The algorithm performs a search over possible sample sizes and identifies one that yields the pre-specified power based  on formula (\ref{power}).

\begin{algorithm}\
\caption{Sample size calculation for $K$-stage design}\label{alg1}
\begin{algorithmic}
\Require \\
\begin{itemize}
\item
$\alpha$: type I error level; alpha spending method; $1-\beta$: theoretical power 
\item
$p_0$: proportion of $n_1$ in $n$; $\left\{p_k\right\}_{k=1}^K$: complete percentages  
\item 
$r$: accrual rate; $t_f$: follow-up duration 
\item 
 $s_0$: value of $S_1(x)$; $s_1$: value of $S_2(x)$
\end{itemize}
\vspace{5pt}
\If{O'Brien-Fleming approach is adopted} 
    \State 
    $\left\{\alpha_k\right\}_{k=1}^K \gets \alpha~\mathrm{and}~\left\{p_k\right\}_{k=1}^K$ 
\EndIf    
\vspace{5pt}
\State {\em Initial power}:~~$p \gets 0$
\State {\em Initial sample size:}~~
$n_0 \gets 100$ (or another small even number)

 \vspace{5pt}
\While {$p<1-\beta$}
     \vspace{5pt}
    \State $n \gets n_0$
    \vspace{5pt}
    \State $n_1 \gets np_0$; $n_2 \gets n(1-p_0)$
    \vspace{5pt}
    \State $\left\{\Sigma_{t_k}\right\}_{k=1}^K \gets ~\text{formula (\ref{variance_log}) with}~n_1, n_2~\text{and}~s_0~\text{under }~H_0$
    \vspace{5pt}
    \State 
    $\boldsymbol{\Sigma}_{K\times K} \gets ~\text{formula (\ref{covariancematrix}) with}~ \left\{\Sigma_{t_k}\right\}_{k=1}^K$ \Comment{covariance matrix under $H_0$}
    \vspace{5pt}
    \If{Lan-DeMets approach is adopted} 
    \State 
    $\left\{\alpha_k\right\}_{k=1}^K \gets \alpha~\text{and}~\left\{\Sigma_{t_k}\right\}_{k=1}^K$ 
    \EndIf   
    \vspace{5pt}
    
    \State 
    $\{c_k\}_{k=1}^K \gets ~\text{formula (\ref{criticalboundaries}) with} \{\alpha_k\}_{k=1}^K \text{~and~} \boldsymbol{\Sigma}_{K\times K}$ 
    \vspace{5pt}
     \State $\left\{\Sigma_{t_k}\right\}_{k=1}^K \gets ~\text{formula (\ref{variance_log}) with}~n_1, n_2, s_0~\text{and}~s_1~\text{under }~H_A$
    \vspace{5pt}
     \State 
     $\boldsymbol{\Sigma}_{K\times K} \gets ~\text{formula (\ref{covariancematrix}) with}~ \left\{\Sigma_{t_k}\right\}_{k=1}^K$ \Comment{covariance matrix under $H_A$}
     \vspace{5pt}
    \State $\boldsymbol{\mu}_{K\times 1} \gets
    ~\text{formula (\ref{mean}) with}~s_0, s_1~\text{and}~\left\{\Sigma_{t_k}\right\}_{k=1}^K$
     \vspace{5pt}
       \State 
    $\left\{\mathrm{RP}_k\right\}_{k=1}^K \gets ~\text{formula (\ref{power}) with}~ \{c_k\}_{k=1}^K \text{~and~}  \boldsymbol{\mu}_{K\times1} \text{~and~} \boldsymbol{\Sigma}_{K\times K}$ 
    \vspace{5pt}
    \State {\em Power at $n$}:~~$p \gets \sum_{k=1}^K\mathrm{RP}_k$
      \vspace{5pt}
     \State $n_0 \gets n_0+2$ 
   \EndWhile
   \State Output sample size $n$ for power $1-\beta$
\end{algorithmic}
\end{algorithm}

\subsection{Comparison with the Binomial Test}\label{sec2.5}   
An alternative approach to our proposed method is the widely used binomial test, which only considers patients with at least follow-up  time $x$. In this section we compare the power of the proposed test with the binomial test theoretically. To facilitate a fair comparison corresponding to the statistic (\ref{statistic}), we prefer a binomial test statistic with a log-log transformation, as follows.
\begin{equation}\label{statistic_binom}
Z_{\mathrm{binom}}(x;t)=\frac{\log\left[-\log\left\{\hat{p}_2(x;t)\right\}\right]-\log\left[-\log\left\{\hat{p}_1(x;t)\right\}\right]}{\sqrt{\frac{1-\hat{p}_1(x;t)}{n_1(x;t)\hat{p}_1(x;t)\left[\log\left\{\hat{p}_1(x;t)\right\}\right]^2}+\frac{1-\hat{p}_2(x;t)}{n_2(x;t)\hat{p}_2(x;t)\left[\log\left\{\hat{p}_2(x;t)\right\}\right]^2}}},
\end{equation}	
where $n_i(x;t)=\sum_{j=1}^{n_i}1_{\left[\tilde{C}_j(t)\ge x\right]}$ denotes the number of patients with at least follow-up  time $x$, at calendar time $t$ in arm $i$ for $i=1, 2$, with $\tilde{C}(t)=\min\{C,(t-Y)^+\}$ representing the overall censoring time variable at calender time $t$, $\hat{p}_i(x;t)$ is the proportion who survive for at least time $x$, among the $n_i(x;t)$ patients.

For the theoretical power, similar to the previous section we rely on a sample-size-dependent  approximation of our proposed statistic~(\ref{statistic}) derived from its asymptotic behavior when $S_1(x)\neq S_2(x)$, as follows.
\begin{align*}
\tilde{Z}(x;t)=Z_{0,1}+ \frac{\log\left[-\log\left\{S_2(x)\right\}\right]-\log\left[-\log\left\{S_1(x)\right\}\right]}
{\sqrt{\Sigma_{t}}},
\end{align*}
where $Z_{0,1}\sim N(0,1)$ and $\Sigma_t$ is as defined by (\ref{variance_log}). Note that $\tilde{Z}(x;t_k)=\tilde{G}_k$ for $\tilde{G}_k$ in (\ref{approx}) when $S_1(x)=s_0$ and $S_2(x)=s_1$.
Similarly, the following sample-size-dependent approximation of the binomial statistic(\ref{statistic_binom}) is also required.
\begin{align*}
\tilde{Z}_{\mathrm{binom}}(x;t)=Z_{0,1}+\frac{\log\left[-\log\left\{S_2(x)\right\}\right]-\log\left[-\log\left\{S_1(x)\right\}\right]}
{\sqrt{\frac{1-S_1(x)}{n_1\Pr\left\{\tilde{C}(t)\ge x\right\}S_1(x)\left[\log\left\{S_1(x)\right\}\right]^2  }+\frac{1-S_2(x)}{n_2\Pr\left\{\tilde{C}(t)\ge x\right\}S_2(x)\left[\log\left\{S_2(x)\right\}\right]^2  }}},
\end{align*}
which follows from $n_i(x; t)/n_i \xrightarrow{p} \Pr\{\tilde{C}(t) \ge x\}$ and $\hat{p}_i(x; t) \xrightarrow{p} S_i(x)$ as $n_i \to \infty$ for $i = 1, 2$.

The following Lemma \ref{lem1} shows that, under the fixed design, the proposed test is theoretically no less powerful than the binomial test. This finding further implies that the proposed test is expected to achieve greater power under the fixed design when certain patients are administratively censored or withdraw from the study before reaching time $x$.
\begin{lem}\label{lem1}
Assume $S_2(x)\neq S_1(x)$, we have
 \begin{align*}
\Pr\left\{\left|\tilde{Z}(x,t)\right|>c\right\}\ge\Pr\left\{\left|\tilde{Z}_{\mathrm{binom}}(x;t)\right|>c\right\},
\end{align*}
where the equality only holds when $\Pr(\tilde{C}(t)\ge x)=1$.
\end{lem}

{\em Proof.}

When $\Pr(\tilde{C}(t)\ge x)<1$, there exists  $v\in(0,x)$, such that $\Pr\left\{\tilde{C}(t)\ge u\right\}> \Pr\left\{\tilde{C}(t) \ge x\right\}$ for any $u$ in $[0,v]$. It follows that
$$
\sigma_i^2(x;t)=\int_0^x\frac{d\Lambda_i(u)}{S_i(u-)\Pr\left\{\tilde{C}(t)\ge u\right\}}
<
\int_0^x\frac{d\Lambda_i(u)}{S_i(u-)\Pr\left\{\tilde{C}(t)\ge x\right\}}
=\frac{1-S_i(x)}{S_i(x)\Pr\left\{\tilde{C}(t)\ge x\right\}}
$$
for $\sigma_i^2(x;t)$ as defined by (\ref{sigma}) with $i=1,2$, where the second equality in the above display is obtained by Stieltjes integral as introduced by~Fleming and Harrington (1991)~\cite{fleming:1991}. Then 
$$
\frac{\sigma_i^2(x;t)}{n_i\left[\log\left\{S_i(x)\right\}\right]^2 }
<\frac{1-S_i(x)}{n_i\Pr\left\{\tilde{C}(t)\ge x\right\}S_i(x)\left[\log\left\{S_i(x)\right\}\right]^2  },
$$
for $i=1,2$. 
Hence, 
\begin{align*}
\frac{\left|\log\left[-\log\left\{S_2(x)\right\}\right]-\log\left[-\log\left\{S_1(x)\right\}\right]\right|}
{\sqrt{\Sigma_t}}=
\frac{\left|\log\left[-\log\left\{S_2(x)\right\}\right]-\log\left[-\log\left\{S_1(x)\right\}\right]\right|}
{\sqrt{\frac{\sigma_1^2(x;t)}{n_1\left[\log\left\{S_1(x)\right\}\right]^2}+\frac{\sigma_2^2(x;t)}{n_2\left[\log\left\{S_2(x)\right\}\right]^2}}}\\
>
\frac{\left|\log\left[-\log\left\{S_2(x)\right\}\right]-\log\left[-\log\left\{S_1(x)\right\}\right]\right|}
{\sqrt{\frac{1-S_1(x)}{n_1\Pr\left\{\tilde{C}(t)\ge x\right\}S_1(x)\left[\log\left\{S_1(x)\right\}\right]^2  }+\frac{1-S_2(x)}{n_2\Pr\left\{\tilde{C}(t)\ge x\right\}S_2(x)\left[\log\left\{S_2(x)\right\}\right]^2  }}}.
\end{align*}
Therefore, 
\begin{align*}
\Pr\left\{\left|\tilde{Z}(x,t)\right|>c\right\}>\Pr\left\{\left|\tilde{Z}_{\mathrm{binom}}(x;t)\right|>c\right\},
\end{align*}
Similarly we see that when $\Pr(\tilde{C}(t)\ge x)=1$,
\begin{align*}
\Pr\left\{\left|\tilde{Z}(x,t)\right|>c\right\}=\Pr\left\{\left|\tilde{Z}_{\mathrm{binom}}(x;t)\right|>c\right\}.~~~~~~~~\square
\end{align*}

\section{Simulation and Analytical Studies}

The simulations compare the empirical type I error and power of our proposed sequential test with the sequential version of the binomial test described in Section~\ref{sec2.5}. In addition, the sequential design is compared with its fixed design counterpart in terms of the number of enrolled patients, based on Algorithm~\ref{alg1}, and the total study duration.

For the proposed sequential test (Section~\ref{sec2.3}), we assume the following parameter settings: a constant accrual rate $r = 120$, follow-up duration $t_f = x$ (the specific time), number of test stages $K = 3$, completion proportions $\{p_1, p_2, p_3\} = \{0.5, 0.75, 1\}$, and nominal type I error $\alpha = 0.05$, with the O’Brien–Fleming approach adopted. All of these parameters are also applied to the sequential binomial test to ensure comparability between the two methods. 

In the first part of the simulation, we evaluate type I error control under the null hypothesis $H_0: S_1(x) = S_2(x) = 0.25$ for $x = 1$ and $x = 5$. We assume that $T_1$ and $T_2$ follow the same Weibull distribution with shape parameter 2. A common loss-to-follow-up censoring time $C \sim \mathrm{Unif}(0, \theta_c)$ is imposed on both groups, where $\theta_c$ is determined using the censoring rate evaluation formula of Halabi and Singh (2004)~\cite{halabi:2004}:

$$\mathrm{CR}=\left(n_1/n\right)\Pr(T_1>C)+\left(n_2/n\right)\Pr(T_2>C),$$
where both $\mathrm{CR}=0.2$ and $0.5$ are considered. Note that $\mathrm{CR}=0$ (the absence of loss-of-follow-up censoring) is also included as one scenario regarding censoring time. 

Ten thousand independent time-to-event datasets are generated for each  total sample size $n=1600, 800, 400$ and $200$, with equal sample sizes  $n_1=n_2=n/2$. For each generated dataset,  both the proposed sequential test and the sequential binomial test are conducted . The empirical type I error for each method is defined as the proportion of rejections among the $10,000$  simulated tests. Figures ~\ref{fig1} and~\ref{fig2} present results for $x=1$ and $5$, respectively. Both figures, constructed from the values in Table~\ref{tab1}, indicate that the two testing approaches achieve adequate control of the cumulative type I error and closely follow the corresponding theoretical cumulative curves.

In the second part of the simulation, we compare the empirical power of the two sequential tests. The data-generating mechanism is similar to that used in the type I error study, except that the distribution of $T_2$ is modified. Specifically, $T_2$ still follows a Weibull distribution with shape parameter 2, but its survival probability at the specific time is set to $S\_2(x) = 0.25 + \mathrm{D}$, where $\mathrm{D} = 0.05, 0.1, 0.15,$ and $0.2$.

For each combination of theoretical power $1-\beta = 80\%$ and $90\%$, censoring rate $\mathrm{CR} = 0, 0.2,$ and $0.5$, and effect size $\mathrm{D} = 0.05, 0.1, 0.15,$ and $0.2$, we employ Algorithm~\ref{alg1} to achieve the required total sample size $n$ with $n_1=n_2=n/2$.  Similar to the preceding type I error study, $10,000$ independent datasets of size $n$ are generated to estimate the empirical power of each sequential test, defined as the proportion of rejections among the corresponding $10,000$ tests. Figures~\ref{fig3} and \ref{fig4} present the simulation results at $x=1$ for $1-\beta=80\%$ and $90\%$, respectively. Figures~\ref{fig5} and \ref{fig6} corresponds to $x=5$. Each panel in these figures corresponds to one scenario defined by the combination of $1-\beta$, $\mathrm{CR}$ and $\mathrm{D}$ and displays the required sample size along with two empirical power curves derived from Tables~\ref{tab2} and~\ref{tab3}. Across all scenarios, the empirical cumulative power closely matches the corresponding theoretical power evaluated by (\ref{power}).\\

It is evident that across all scenarios, the empirical cumulative power values are very close to their theoretical counterparts evaluated by (\ref{power}). In terms of testing performance, the proposed sequential test consistently outperforms  the sequential binomial test, even when only administrative censoring is present (at the first two monitoring looks with $\mathrm{CR}=0$).  Moreover, the advantage of the proposed method becomes substantially more pronounced in the presence of loss-to-follow-up censoring. This finding is consistent with the theoretical result in Lemma~\ref{lem1}, which establishes that the proposed sequential test is more powerful than the binomial approach.  \\

The supplementary document also presents results for the remaining scenarios: Weibull with Lan–DeMets boundaries (Figures S1–S2: empirical type I error; Figures S3–S6: empirical power), exponential with O’Brien–Fleming boundaries (Figures S7–S8: empirical type I error; Figures S9–S12: empirical power), and exponential with Lan–DeMets boundaries (Figures S13–S14: empirical type I error; Figures S15–S18: empirical power). These results further demonstrate the advantages of the proposed method over the binomial approach.

\begin{figure}
\begin{center}
			\scalebox{0.7}{\includegraphics[angle =0]{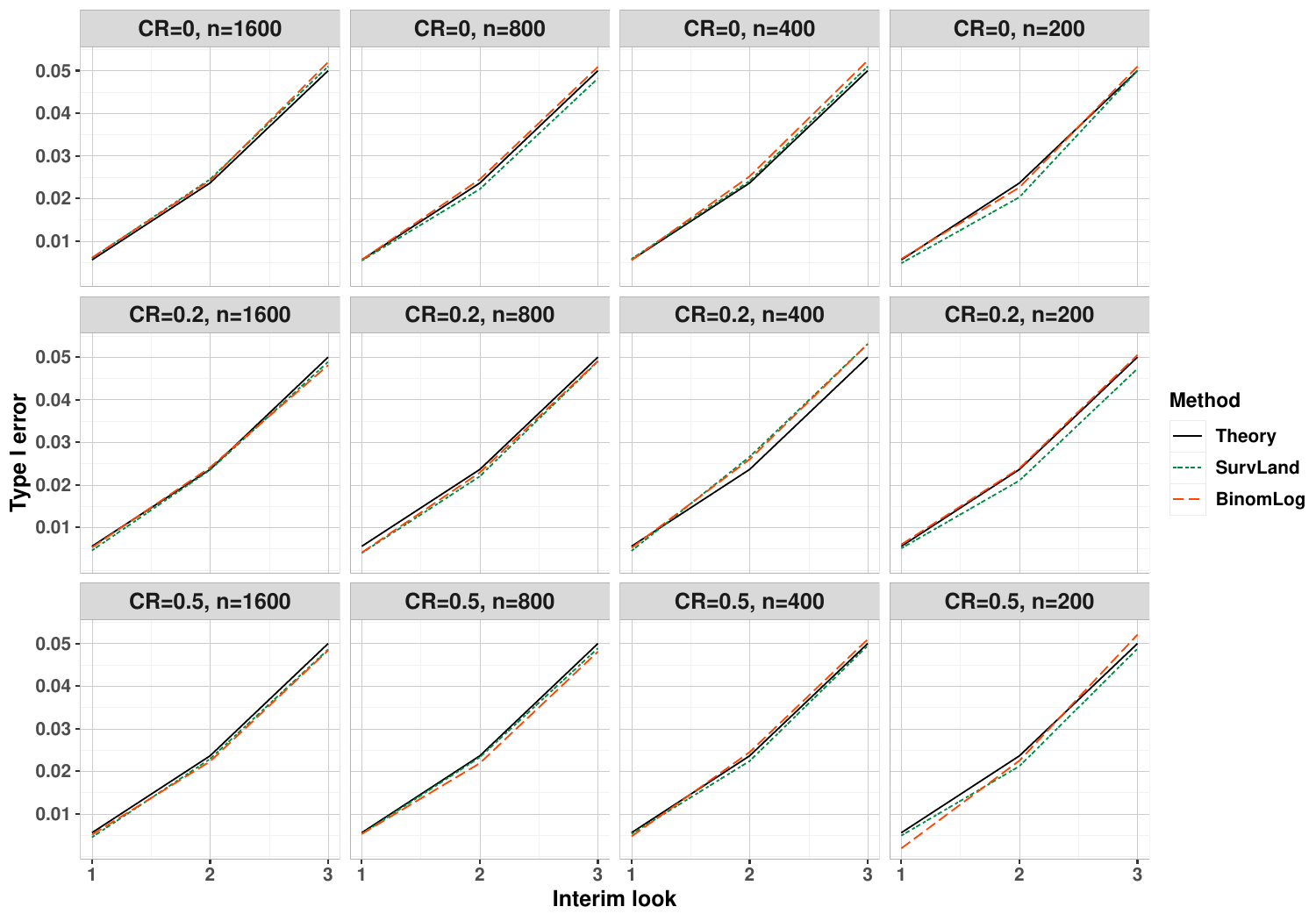}}
\caption{Type I error comparisons between proposed test (SurvLand) and binomial test (BinomLog) based on  survival rate at $x=1$.
}\label{fig1}
\end{center}
\end{figure}

\begin{figure}
\begin{center}
			\scalebox{0.7}{\includegraphics[angle =0]{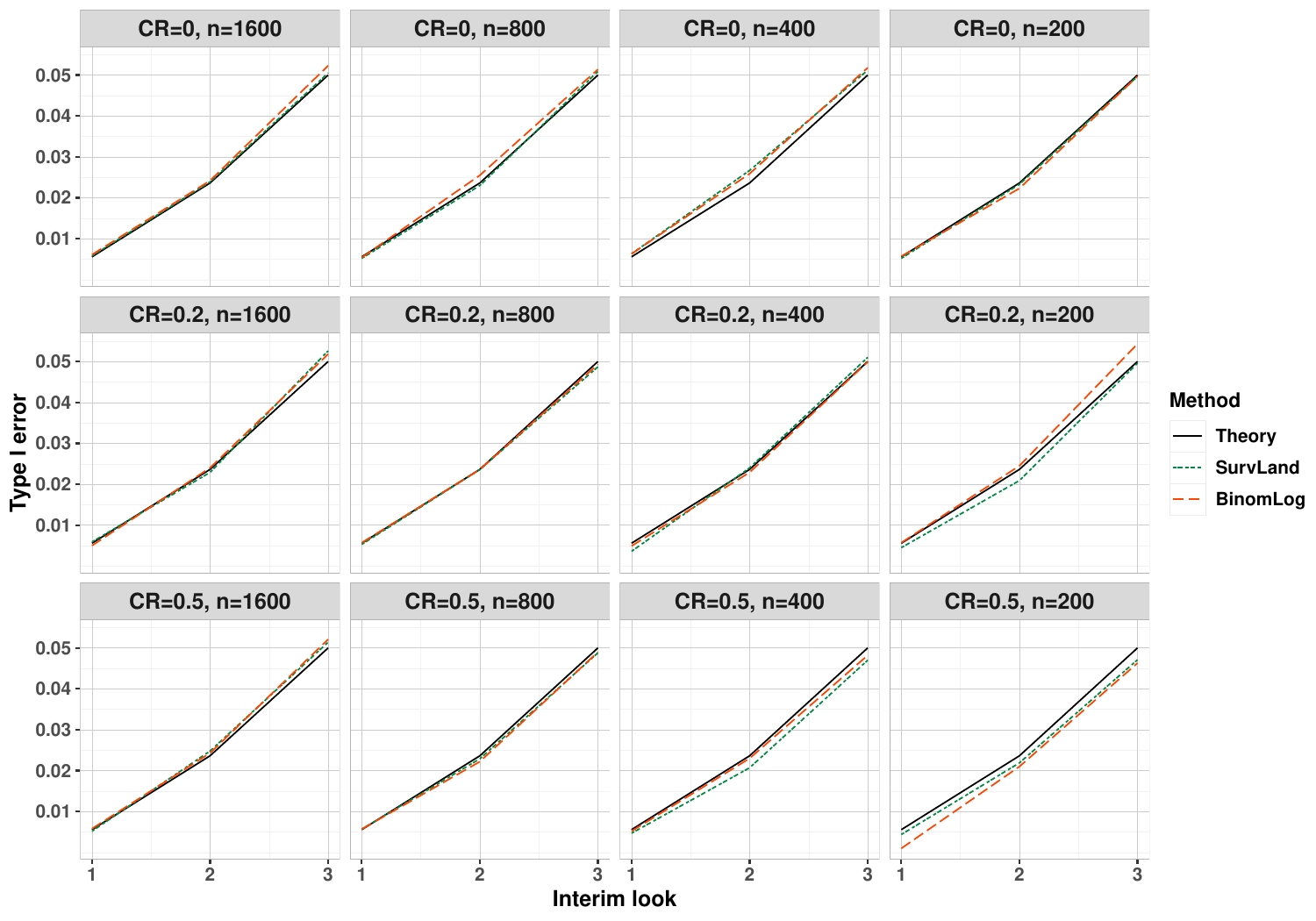}}
\caption{Type I error comparisons between proposed test (SurvLand) and binomial test  (BinomLog) based on survival rate at $x=5$}\label{fig2}
\end{center}
\end{figure}

\begin{table}[htp]
\centering\caption{Empirical type I error comparisons between SurvLand and BinomLog}\label{tab1}
\scriptsize
\begin{tabular}{rrrrrrrrrrrrr}
\hline
&\multicolumn{12}{c}{Information}\\
&$\mathrm{50\%}$&$\mathrm{75\%}$&$\mathrm{100\%}$&
$\mathrm{50\%}$&$\mathrm{75\%}$&$\mathrm{100\%}$&$\mathrm{50\%}$&$\mathrm{75\%}$&$\mathrm{100\%}$&
$\mathrm{50\%}$&$\mathrm{75\%}$&$\mathrm{100\%}$\\
\cmidrule(lr){2-4}\cmidrule(lr){5-7}\cmidrule(lr){8-10}\cmidrule(lr){11-13}
&\multicolumn{3}{c}{$\mathrm{n=1600}$}& \multicolumn{3}{c}{$\mathrm{n=800}$}&\multicolumn{3}{c}{$\mathrm{n=400}$}& \multicolumn{3}{c}{$\mathrm{n=200}$}\\
\hline
&\multicolumn{12}{c}{$\text{Survival rate at } x=1$}\\
\multicolumn{1}{r}{$\mathrm{CR=0}$}&&&&&&&&&&&&\\
$\mathrm{SurvLand}$&0.006	&0.024	&0.052
&0.006	&0.025	&0.051
&0.006	&0.025	&0.052
&0.006	&0.023	&0.051\\
$\mathrm{BinomLog}$&0.006	&0.024	&0.053
	&0.006	&0.026	&0.053
&0.005	&0.024	&0.049
	&0.002	&0.020	&0.048\\
\cmidrule(lr){2-4}\cmidrule(lr){5-7}\cmidrule(lr){8-10}\cmidrule(lr){11-13}
\multicolumn{1}{r}{$\mathrm{CR=0.2}$}&&&&&&&&&&&&\\
$\mathrm{SurvLand}$&0.005	&0.024	&0.049
&0.004	&0.022	&0.049
&0.005	&0.027	&0.053
&0.005	&0.021	&0.047\\
$\mathrm{BinomLog}$&0.005	&0.024	&0.048
	&0.004	&0.023	&0.049
&0.005	&0.026	&0.053
	&0.006	&0.024	&0.051\\
\cmidrule(lr){2-4}\cmidrule(lr){5-7}\cmidrule(lr){8-10}\cmidrule(lr){11-13}
\multicolumn{1}{r}{$\mathrm{CR=0.5}$}&&&&&&&&&&&&\\
$\mathrm{SurvLand}$&0.005	&0.023	&0.049
&0.005	&0.023	&0.049
&0.005	&0.022	&0.050
&0.005	&0.021	&0.049\\
$\mathrm{BinomLog}$&0.005	&0.022	&0.048
	&0.005	&0.022	&0.048
	&0.005	&0.025	&0.051
	&0.002	&0.023	&0.052\\
\cmidrule(lr){2-4}\cmidrule(lr){5-7}\cmidrule(lr){8-10}\cmidrule(lr){11-13}
\hline
&\multicolumn{12}{c}{$\text{Survival rate at } x=5$}\\
\multicolumn{1}{r}{$\mathrm{CR=0}$}&&&&&&&&&&&&\\
$\mathrm{SurvLand}$&0.006	&0.024	&0.051
&0.005	&0.023	&0.051
&0.006	&0.027	&0.051
&0.005	&0.023	&0.050\\
$\mathrm{BinomLog}$&0.006	&0.024	&0.052
	&0.005	&0.026	&0.051
&0.006	&0.026	&0.052
	&0.006	&0.022	&0.050\\
\cmidrule(lr){2-4}\cmidrule(lr){5-7}\cmidrule(lr){8-10}\cmidrule(lr){11-13}
\multicolumn{1}{r}{$\mathrm{CR=0.2}$}&&&&&&&&&&&&\\
$\mathrm{SurvLand}$&0.006	&0.023	&0.053
&0.005	&0.024	&0.049
&0.004	&0.024	&0.051
&0.005	&0.021	&0.050\\
$\mathrm{BinomLog}$&0.005	&0.024	&0.052
	&0.006	&0.024	&0.049
&0.005	&0.023	&0.050
	&0.006	&0.025	&0.054\\
\cmidrule(lr){2-4}\cmidrule(lr){5-7}\cmidrule(lr){8-10}\cmidrule(lr){11-13}
\multicolumn{1}{r}{$\mathrm{CR=0.5}$}&&&&&&&&&&&&\\
$\mathrm{SurvLand}$&0.005	&0.025	&0.051
&0.006	&0.023	&0.049
&0.005	&0.021	&0.047
&0.004	&0.022	&0.047\\
$\mathrm{BinomLog}$&0.006	&0.024	&0.052
	&0.006	&0.022	&0.049
&0.005	&0.023	&0.048
	&0.001	&0.021	&0.046\\
\cmidrule(lr){2-4}\cmidrule(lr){5-7}\cmidrule(lr){8-10}\cmidrule(lr){11-13}
\hline
\end{tabular}
\end{table}

\begin{figure}
\begin{center}
			\scalebox{0.7}{\includegraphics[angle =0]{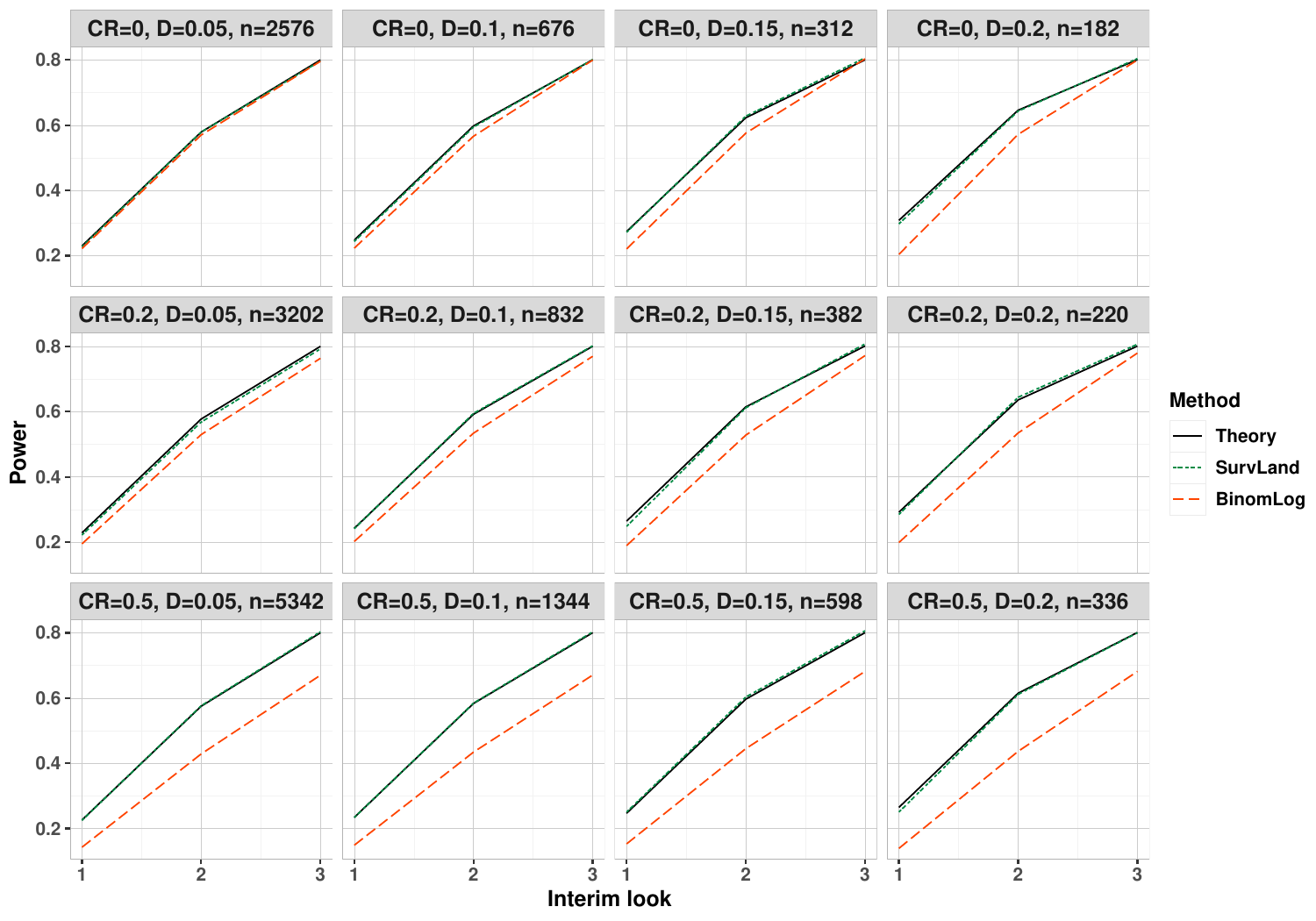}}
\caption{Empirical power comparisons between  proposed test (SurvLand) and  binomial test (BinomLog) with  theoretical power of 80\% based on survival rate at $x=1$ }\label{fig3}
\end{center}
\end{figure}

\begin{figure}
\begin{center}
			\scalebox{0.7}{\includegraphics[angle =0]{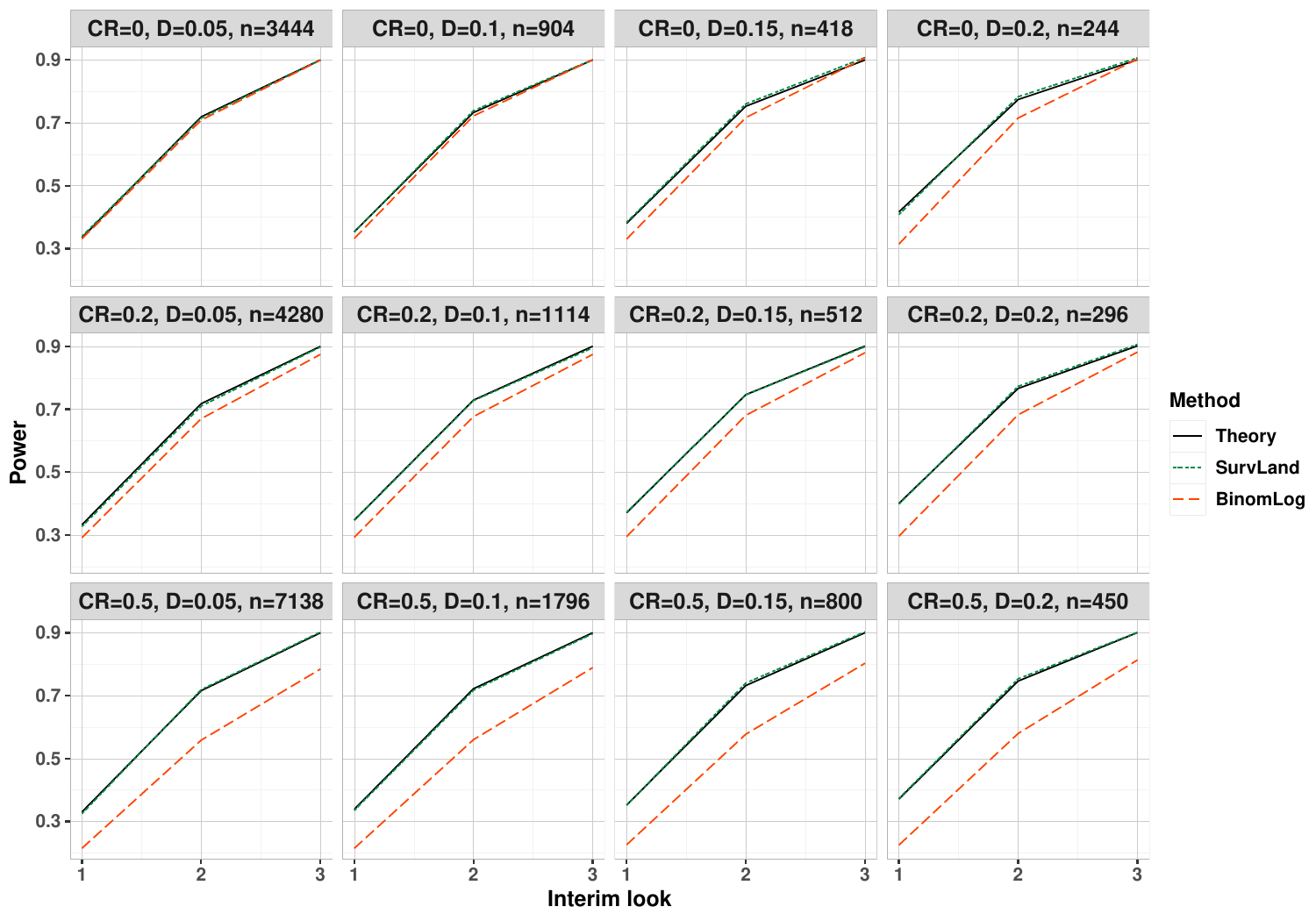}}
\caption{Empirical power comparisons between proposed test (SurvLand) and  binomial test (BinomLog) with  theoretical power of 90\% based on survival rate at $x=1$}\label{fig4}
\end{center}
\end{figure}

\begin{table}[htp]
\centering\caption{Empirical power comparisons between SurvLand and BinomLog based on survival rate at $x=1$}\label{tab2}
\scriptsize
\begin{tabular}{rrrrrrrrrrrrr}
\hline
&\multicolumn{12}{c}{Information}\\
&$\mathrm{50\%}$&$\mathrm{75\%}$&$\mathrm{100\%}$&
$\mathrm{50\%}$&$\mathrm{75\%}$&$\mathrm{100\%}$&$\mathrm{50\%}$&$\mathrm{75\%}$&$\mathrm{100\%}$&
$\mathrm{50\%}$&$\mathrm{75\%}$&$\mathrm{100\%}$\\
\cmidrule(lr){2-4}\cmidrule(lr){5-7}\cmidrule(lr){8-10}\cmidrule(lr){11-13}
&\multicolumn{3}{c}{$\mathrm{D=0.05}$}& \multicolumn{3}{c}{$\mathrm{D=0.1}$}&\multicolumn{3}{c}{$\mathrm{D=0.15}$}& \multicolumn{3}{c}{$\mathrm{D=0.2}$}\\
\hline
&\multicolumn{12}{c}{$\mathrm{Thoretical~power=80\%}$}\\
\multicolumn{1}{r}{$\mathrm{CR=0}$}&&&&&&&&&&&&\\
$\mathrm{SurvLand}$&0.227	&0.579	&0.797
&0.244	&0.594	&0.801
&0.272	&0.628	&0.807
&0.297	&0.643	&0.804\\
$\mathrm{BinomLog}$&0.222&	0.570&	0.796
	&0.223&	0.566	&0.800
&0.221	&0.576&	0.804
	&0.204&	0.572&	0.800\\
\cmidrule(lr){2-4}\cmidrule(lr){5-7}\cmidrule(lr){8-10}\cmidrule(lr){11-13}
\multicolumn{1}{r}{$\mathrm{CR=0.2}$}&&&&&&&&&&&&\\
$\mathrm{SurvLand}$&0.222	&0.568	&0.792
&0.241	&0.595	&0.802
&0.248	&0.612	&0.808
&0.276	&0.645	&0.810\\
$\mathrm{BinomLog}$&0.195	&0.530&	0.763
	&0.202	&0.534&	0.770
&0.189	&0.528	&0.772
	&0.199	&0.535	&0.780\\
\cmidrule(lr){2-4}\cmidrule(lr){5-7}\cmidrule(lr){8-10}\cmidrule(lr){11-13}
\multicolumn{1}{r}{$\mathrm{CR=0.5}$}&&&&&&&&&&&&\\
$\mathrm{SurvLand}$&0.225	&0.576	&0.803
&0.234	&0.584	&0.803
&0.250	&0.603	&0.806
&0.250	&0.611	&0.800\\
$\mathrm{BinomLog}$&0.143	&0.429	&0.670
	&0.149	&0.434	&0.671
	&0.153	&0.445	&0.682
	&0.139	&0.437	&0.682\\
\cmidrule(lr){2-4}\cmidrule(lr){5-7}\cmidrule(lr){8-10}\cmidrule(lr){11-13}
\hline
&\multicolumn{12}{c}{$\mathrm{Thoretical~power=90\%}$}\\
\multicolumn{1}{r}{$\mathrm{CR=0}$}&&&&&&&&&&&&\\
$\mathrm{SurvLand}$&0.339	&0.715	&0.900
&0.353	&0.739	&0.901
&0.383	&0.760	&0.908
&0.408	&0.783	&0.907\\
$\mathrm{BinomLog}$&0.331	&0.709	&0.899
	&0.332	&0.722	&0.901
&0.330	&0.717	&0.906
	&0.314	&0.715	&0.903\\
\cmidrule(lr){2-4}\cmidrule(lr){5-7}\cmidrule(lr){8-10}\cmidrule(lr){11-13}
\multicolumn{1}{r}{$\mathrm{CR=0.2}$}&&&&&&&&&&&&\\
$\mathrm{SurvLand}$&0.327	&0.710	&0.898
&0.347	&0.728	&0.896
&0.372	&0.746	&0.899
&0.398	&0.773	&0.906\\
$\mathrm{BinomLog}$&0.292	&0.670	&0.874
	&0.293	&0.677	&0.875
&0.295	&0.681	&0.880
	&0.296	&0.682	&0.882\\
\cmidrule(lr){2-4}\cmidrule(lr){5-7}\cmidrule(lr){8-10}\cmidrule(lr){11-13}
\multicolumn{1}{r}{$\mathrm{CR=0.5}$}&&&&&&&&&&&&\\
$\mathrm{SurvLand}$&0.325	&0.719	&0.903
&0.335	&0.717	&0.897
&0.351	&0.740	&0.903
&0.373	&0.754	&0.902\\
$\mathrm{BinomLog}$&0.215	&0.558	&0.785
	&0.214	&0.560	&0.789
&0.226	&0.577	&0.803
	&0.225	&0.580	&0.813\\
\cmidrule(lr){2-4}\cmidrule(lr){5-7}\cmidrule(lr){8-10}\cmidrule(lr){11-13}
\hline
\end{tabular}
\end{table}

\begin{figure}
\begin{center}
			\scalebox{0.7}{\includegraphics[angle =0]{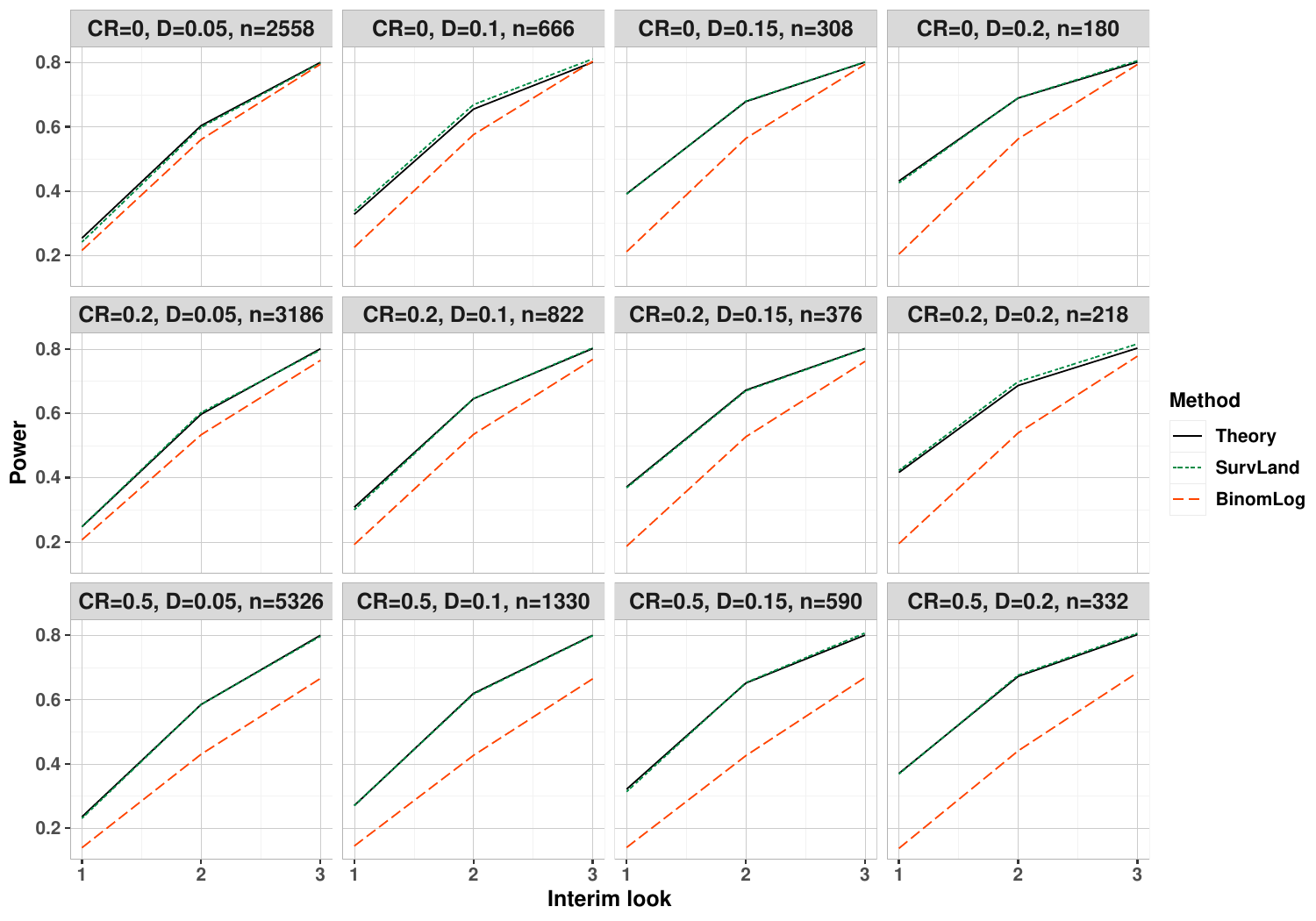}}
\caption{Empirical power comparisons between proposed test (SurvLand) and  binomial test (BinomLog) with theoretical power of 80\% based on survival rate at $x=5$ }\label{fig5}
\end{center}
\end{figure}

\begin{figure}
\begin{center}
			\scalebox{0.7}{\includegraphics[angle =0]{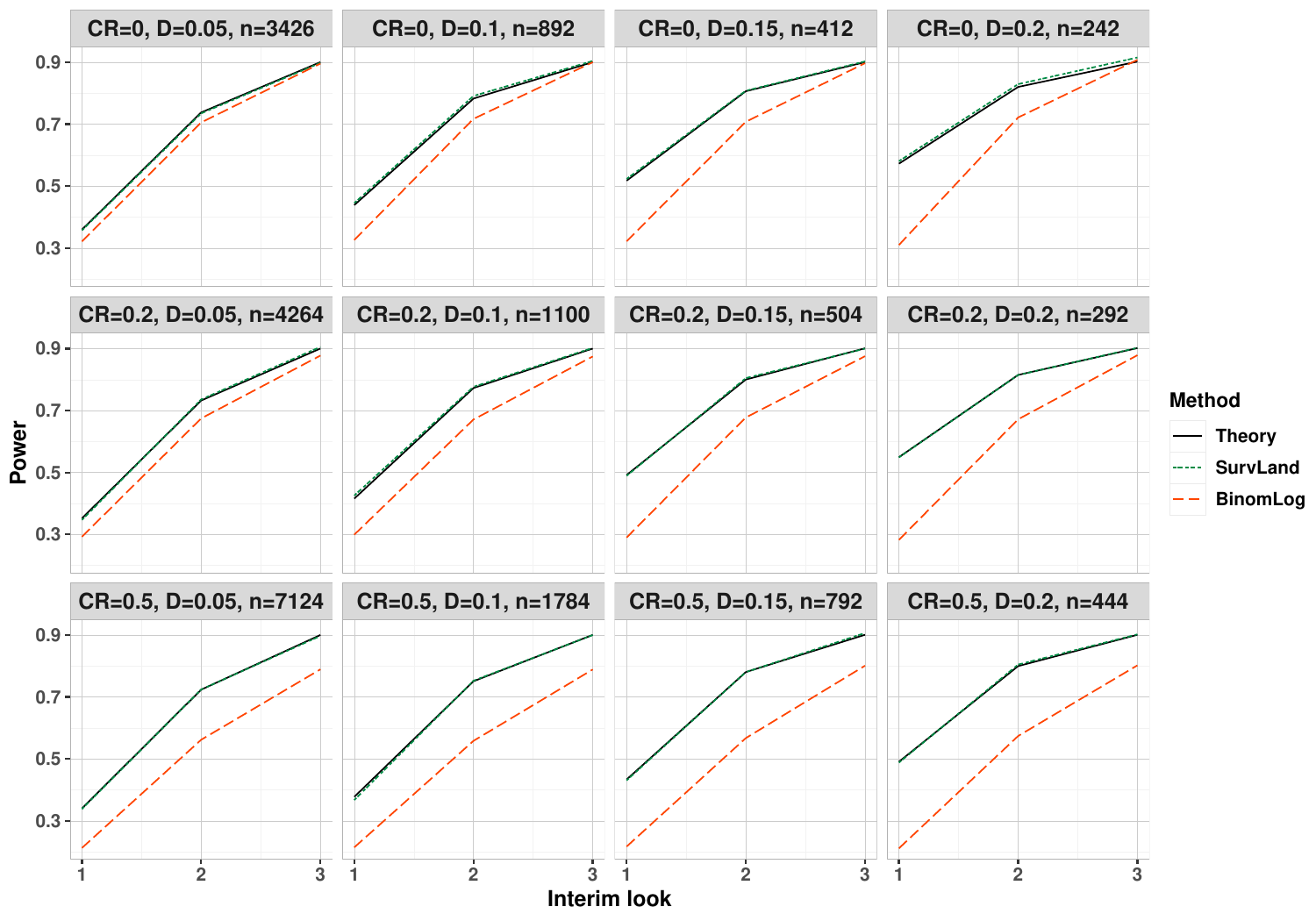}}
\caption{Empirical power comparisons between proposed test (SurvLand) and  binomial test (BinomLog) with theoretical power of 90\% based on survival rate at $x=5$}\label{fig6}
\end{center}
\end{figure}

\begin{table}[htp]
\centering\caption{Empirical power comparisons between SurvLand and BinomLog based on survival rate at $x=5$}\label{tab3}
\scriptsize
\begin{tabular}{rrrrrrrrrrrrr}
\hline
&\multicolumn{12}{c}{Information}\\
&$\mathrm{50\%}$&$\mathrm{75\%}$&$\mathrm{100\%}$&
$\mathrm{50\%}$&$\mathrm{75\%}$&$\mathrm{100\%}$&$\mathrm{50\%}$&$\mathrm{75\%}$&$\mathrm{100\%}$&
$\mathrm{50\%}$&$\mathrm{75\%}$&$\mathrm{100\%}$\\
\cmidrule(lr){2-4}\cmidrule(lr){5-7}\cmidrule(lr){8-10}\cmidrule(lr){11-13}
&\multicolumn{3}{c}{$\mathrm{D=0.05}$}& \multicolumn{3}{c}{$\mathrm{D=0.1}$}&\multicolumn{3}{c}{$\mathrm{D=0.15}$}& \multicolumn{3}{c}{$\mathrm{D=0.2}$}\\
\hline
&\multicolumn{12}{c}{$\mathrm{Thoretical~power=80\%}$}\\
\multicolumn{1}{r}{$\mathrm{CR=0}$}&&&&&&&&&&&&\\
$\mathrm{SurvLand}$&0.242	&0.599	&0.798
&0.338	&0.669	&0.811
&0.391	&0.680	&0.802
&0.426	&0.690	&0.805\\
$\mathrm{BinomLog}$&0.216	&0.560	&0.795
	&0.225	&0.576	&0.803
&0.212	&0.564	&0.794
	&0.204	&0.562	&0.794\\
\cmidrule(lr){2-4}\cmidrule(lr){5-7}\cmidrule(lr){8-10}\cmidrule(lr){11-13}
\multicolumn{1}{r}{$\mathrm{CR=0.2}$}&&&&&&&&&&&&\\
$\mathrm{SurvLand}$&0.248	&0.602	&0.797
&0.300	&0.645	&0.803
&0.368	&0.669	&0.800
&0.422	&0.698	&0.815\\
$\mathrm{BinomLog}$&0.206	&0.533	&0.764
	&0.192	&0.534	&0.767
&0.186	&0.527	&0.761
	&0.195	&0.539	&0.777\\
\cmidrule(lr){2-4}\cmidrule(lr){5-7}\cmidrule(lr){8-10}\cmidrule(lr){11-13}
\multicolumn{1}{r}{$\mathrm{CR=0.5}$}&&&&&&&&&&&&\\
$\mathrm{SurvLand}$&0.231	&0.585	&0.797
&0.271	&0.617	&0.799
&0.314	&0.653	&0.807
&0.369	&0.676	&0.806\\
$\mathrm{BinomLog}$&0.140	&0.4303	&0.666
	&0.145	&0.427	&0.665
	&0.140	&0.426	&0.669
	&0.138	&0.441	&0.684\\
\cmidrule(lr){2-4}\cmidrule(lr){5-7}\cmidrule(lr){8-10}\cmidrule(lr){11-13}
\hline
&\multicolumn{12}{c}{$\mathrm{Thoretical~power=90\%}$}\\
\multicolumn{1}{r}{$\mathrm{CR=0}$}&&&&&&&&&&&&\\
$\mathrm{SurvLand}$&0.357	&0.735	&0.898
&0.445	&0.791	&0.904
&0.523	&0.808	&0.902
&0.581	&0.829	&0.915\\
$\mathrm{BinomLog}$&0.321	&0.705	&0.896
	&0.326	&0.717	&0.900
&0.322	&0.707	&0.897
	&0.31	&0.722	&0.907\\
\cmidrule(lr){2-4}\cmidrule(lr){5-7}\cmidrule(lr){8-10}\cmidrule(lr){11-13}
\multicolumn{1}{r}{$\mathrm{CR=0.2}$}&&&&&&&&&&&&\\
$\mathrm{SurvLand}$&0.347	&0.736	&0.906
&0.426	&0.777	&0.902
&0.490	&0.805	&0.900
&0.549	&0.815	&0.902\\
$\mathrm{BinomLog}$&0.293	&0.675	&0.877
	&0.299	&0.671	&0.875
&0.290	&0.678	&0.875
	&0.283	&0.672	&0.879\\
\cmidrule(lr){2-4}\cmidrule(lr){5-7}\cmidrule(lr){8-10}\cmidrule(lr){11-13}
\multicolumn{1}{r}{$\mathrm{CR=0.5}$}&&&&&&&&&&&&\\
$\mathrm{SurvLand}$&0.338	&0.724	&0.897
&0.367	&0.753	&0.900
&0.431	&0.780	&0.906
&0.488	&0.805	&0.902\\
$\mathrm{BinomLog}$&0.213	&0.562	&0.789
	&0.215	&0.558	&0.789
&0.217	&0.567	&0.801
	&0.212	&0.574	&0.802\\
\cmidrule(lr){2-4}\cmidrule(lr){5-7}\cmidrule(lr){8-10}\cmidrule(lr){11-13}
\hline
\end{tabular}
\end{table}

On the other hand, we assess the added value of the proposed sequential test relative to the corresponding fixed-design test using the sample sizes obtained from the preceding power study, rather than from simulated data. Tables~\ref{tab4} and~\ref{tab5} compare the two designs in terms of the required number of enrolled patients and the total study duration, including both accrual and follow-up periods. Table~\ref{tab4} presents results for the specific time point $x=1$, whereas Table~\ref{tab5} reports results for $x=5$. In both tables, the number of enrolled patients and the study duration are reported for each combination of $1-\beta$, $\mathrm{CR}$, and $\mathrm{D}$. For the sequential test, the expected number of enrolled patients and the expected study duration—weighted by power across monitoring looks—are used for comparison.

The proposed sequential design generally yields shorter study durations than the corresponding fixed design for both specific times $x=1$ and $x=5$ (Tables~\ref{tab4} and~\ref{tab5}). However, its ability to reduce the required sample size is more pronounced for $x=1$ than for $x=5$. This occurs because, when the landmark is $x=5$, an interim monitoring time of the form $t = 5 +$ (a partial enrollment period) may exceed the total enrollment duration in some settings, thereby eliminating opportunities for early stopping and consequent sample size reduction. For instance, when the total sample size is below 600, the enrollment duration is less than 5 under an accrual rate of $r=120$, implying that all patients have already been enrolled by the first monitoring time $t>5$.

\begin{table}[htp]
\centering\caption{Enrolled subject number and study duration comparisons between  proposed sequential design and corresponding fixed design based on survival rate at $x=1$}\label{tab4}
\scriptsize
\begin{tabular}{rcccccccc}
\hline
&\multicolumn{8}{c}{$\mathrm{Survival~~rate~~difference}$}\\
&0.05&0.1&0.15&0.2&0.05&0.1&0.15&0.2\\
\cmidrule(lr){2-5}\cmidrule(lr){6-9}
&\multicolumn{4}{c}{$\mathrm{Enrolled~~number}$}&\multicolumn{4}{c}{$\mathrm{Study~~duration}$}\\
\hline
&\multicolumn{8}{c}{$\mathrm{Thoretical~power=80\%}$}\\
\multicolumn{1}{r}{$\mathrm{CR=0}$} & & & & & & & &\\
$\mathrm{Sequential}$&2126&606&304&182&18.13&5.44&3.02&2.15\\
$\mathrm{Fixed}$&2504&660&308&180&21.87&6.50&3.57&2.50\\
\cmidrule(lr){2-5}\cmidrule(lr){6-9}
\multicolumn{1}{r}{$\mathrm{CR=0.2}$} & & & & & & & &\\
$\mathrm{Sequential}$&2628&730&364&220&22.31&6.49&3.48&2.41\\
$\mathrm{Fixed}$&3112&812&374&218&26.93&7.77&4.12&2.82\\
\cmidrule(lr){2-5}\cmidrule(lr){6-9}
\multicolumn{1}{r}{$\mathrm{CR=0.5}$} & & & & & & & &\\
$\mathrm{Sequential}$&4342&1140&544&324&36.60&9.91&4.93&3.18\\
$\mathrm{Fixed}$&5186&1308&584&330&44.22&11.90&5.87&3.75\\
\cmidrule(lr){2-5}\cmidrule(lr){6-9}
\hline
&\multicolumn{8}{c}{$\mathrm{Thoretical~power=90\%}$}\\
\multicolumn{1}{r}{$\mathrm{CR=0}$} & & & & & & & &\\
$\mathrm{Sequential}$&2624&748&386&244&22.14&6.49&3.50&2.43\\
$\mathrm{Fixed}$&3352&884&412&242&28.93&8.37&4.43&3.02\\
\cmidrule(lr){2-5}\cmidrule(lr){6-9}
\multicolumn{1}{r}{$\mathrm{CR=0.2}$} & & & & & & & &\\
$\mathrm{Sequential}$&3244&902&460&286&27.30&7.78&4.07&2.75\\
$\mathrm{Fixed}$&4166&1086&500&290&35.72&10.05&5.17&3.42\\
\cmidrule(lr){2-5}\cmidrule(lr){6-9}
\multicolumn{1}{r}{$\mathrm{CR=0.5}$} & & & & & & & &\\
$\mathrm{Sequential}$&5358&1408&672&421&44.93&12.00&5.86&3.70\\
$\mathrm{Fixed}$&6942&1750&782&440&58.85&15.58&7.52&4.67\\
\cmidrule(lr){2-5}\cmidrule(lr){6-9}
\hline
\end{tabular}
\end{table}

\begin{table}[htp]
\centering\caption{Enrolled subject number and study duration comparisons between proposed sequential design and corresponding fixed design based on survival rate at $x=5$}\label{tab5}
\scriptsize
\begin{tabular}{rcccccccc}
\hline
&\multicolumn{8}{c}{$\mathrm{Survival~~rate~~difference}$}\\
&0.05&0.1&0.15&0.2&0.05&0.1&0.15&0.2\\
\cmidrule(lr){2-5}\cmidrule(lr){6-9}
&\multicolumn{4}{c}{$\mathrm{Enrolled~~number}$}&\multicolumn{4}{c}{$\mathrm{Study~~duration}$}\\
\hline
&\multicolumn{8}{c}{$\mathrm{Thoretical~power=80\%}$}\\
\multicolumn{1}{r}{$\mathrm{CR=0}$} & & & & & & & &\\
$\mathrm{Sequential}$&2372&666&308&180&21.75&9.19&6.88&6.08\\
$\mathrm{Fixed}$&2504&660&308&180&25.87&10.50&7.57&6.50\\
\cmidrule(lr){2-5}\cmidrule(lr){6-9}
\multicolumn{1}{r}{$\mathrm{CR=0.2}$} & & & & & & & &\\
$\mathrm{Sequential}$&2872&822&376&218&25.95&10.22&7.32&6.32\\
$\mathrm{Fixed}$&3112&812&374&218&30.93&11.77&8.12&6.82\\
\cmidrule(lr){2-5}\cmidrule(lr){6-9}
\multicolumn{1}{r}{$\mathrm{CR=0.5}$} & & & & & & & &\\
$\mathrm{Sequential}$&4586&1314&590&332&40.28&13.62&8.72&7.05\\
$\mathrm{Fixed}$&5186&1308&584&330&48.22&15.90&9.87&7.75\\
\cmidrule(lr){2-5}\cmidrule(lr){6-9}
\hline
&\multicolumn{8}{c}{$\mathrm{Thoretical~power=90\%}$}\\
\multicolumn{1}{r}{$\mathrm{CR=0}$} & & & & & & & &\\
$\mathrm{Sequential}$&2930&892&412&242&25.72&10.16&7.30&6.31\\
$\mathrm{Fixed}$&3352&884&412&242&32.93&12.37&8.43&7.02\\
\cmidrule(lr){2-5}\cmidrule(lr){6-9}
\multicolumn{1}{r}{$\mathrm{CR=0.2}$} & & & & & & & &\\
$\mathrm{Sequential}$&3548&1100&504&292&30.90&11.44&7.84&6.60\\
$\mathrm{Fixed}$&4166&1086&500&290&39.72&14.05&9.17&7.42\\
\cmidrule(lr){2-5}\cmidrule(lr){6-9}
\multicolumn{1}{r}{$\mathrm{CR=0.5}$} & & & & & & & &\\
$\mathrm{Sequential}$&5664&1674&792&444&48.58&15.67&9.60&7.51\\
$\mathrm{Fixed}$&6942&1750&782&440&62.85&19.58&11.52&8.67\\
\cmidrule(lr){2-5}\cmidrule(lr){6-9}
\hline
\end{tabular}
\end{table}


\section{Real Life Examples}
\subsection{Renal Mass Example}
Using our shiny application \href{https://ynu0gt-chenxi-yu.shinyapps.io/shinyapp/}
{\textcolor{blue}{Sample Size/Power Calculation of Group Sequential Design in Phase III Clinical Trials}} with the provided motivating example, we assume a two-sided type I error of 0.05, power of 0.85, an IFS rate at 5-years of 25\% in the control arm, equal allocation between the two arms, a hypothesized 10\% difference between the two arms, exponentially distributed even times, and three analyses planned at at 50\%, 75\% and 100\% total information. Inputting all these parameters in the R shiny application, the average required sample size is 756 patients. On the other hand, if we use the binomial approach, 770 patients on average are needed to detect the anticipated difference.  Although the average sample sizes are calculated as the power weighted sums of sample sizes across interim analyses, both average numbers are actually the total sample sizes, since all patients have already enrolled at the first interim analysis point. Our method would be able to reduce sample size more substantially for a shorter time. We generate a synthetic dataset based on the experimental design settings to simulate the study. At the first interim analysis, the test statistic 
$|Z(5, t_1)| = 1.842 < 2.772 = c_1$, which does not meet the threshold for rejection. However, at the second interim look, $|Z(5, t_2)| = 2.568 > 2.280 = c_2$, 
consequently, the null hypothesis is rejected.

\subsection{Prostate Cancer Example}
We are interested in designing a randomized phase II trial where metastatic prostate cancer patients will be randomized to an androgen receptor pathway inhibitor (ARPI, standard of care) or a radio-ligand. The primary endpoint is radiographic rPFS at 12 months, which is defined as the interval from the date of random assignment to date of radiographic progression or death, whichever occurs first. Based on observed data, the rPFS rate at 12 months in the standard of care arm is 10\%. The hypothesized  difference between the two arms is 15\%. We assume a 2:1 allocation ratio  between the SOC vs. the experimental arm, a one-sided type I error rate of 0.10, accrual rate of 3 patients/month, and the follow-up period of 12 months. A power of 0.90 is desired and two analyses are planned at 50\% and 100\% total information. 
Using our shiny application program, we  enter the above values and choose the exponential distribution, the program yields a average sample size of 161.5 patients. Keeping the same assumptions as above, but using the binomial test between two arms, the target average sample size is 164.1.  Following the  the preceding example, we generate a synthetic dataset to simulate the study. The null hypothesis is rejected at the first interim analysis, as the test statistic $Z(12, t_1) = 2.384$ exceeds the critical value $c_1 = 2.054$. \\

\section{Discussion}
In this article, we present a novel approach for sample size determination and type I error control in both fixed and sequential designs. Milestone survival estimates the probability of remaining event-free at a prespecified time, incorporating all follow-up data up to that point \cite{Chen:2015}. Unlike landmark analysis \cite{Anderson:1983}, which excludes patients who do not survive to the chosen time, milestone survival provides an unbiased estimate of treatment effect while retaining early events. This approach facilitates standardized comparisons across arms, supports meta-analytic evaluations, and enables earlier assessment of efficacy without discarding valuable data.\\

Most investigators use the binomial test when designing trials that compare milestone survival between two groups, owing to its simplicity and ease of implementation. This approach effectively reduces the time-to-event outcome to a binary indicator at a fixed time point and relies on the assumption that censoring prior to the milestone time is independent and non-informative. In practice, this assumption is often unrealistic, particularly when censoring occurs before the milestone or differs between treatment arms. As a result, binomial-based analyses may suffer from loss of efficiency, reduced power, or biased inference compared with methods that incorporate all available follow-up information, such as the approach proposed here.\\

Using real life examples in the design of a renal and a prostate cancer  trials, we show that our method maintains type I error near the nominal level and achieves lower sample sizes than the conventional binomial test. Simulations across different hypothesized effect sizes, censoring rates, and nominal powers confirmed consistent performance for both fixed and sequential designs. Across multiple interim looks, our method consistently produced smaller type I error and higher power, regardless of the underlying failure time distribution or error spending function, demonstrating robustness and practical utility.\\

Our approach is equally applicable to randomized phase II trials, which often involve smaller sample sizes and exploratory endpoints \cite{neal:2021}. It is also well-suited for immunotherapy trials or other settings where non-proportional hazards are expected between treatment arms \cite{{Mukhopadhyay:2022},{Apolo:2025}}, supporting reliable study design even when treatment effects vary over time, delayed responses occur, or survival curves cross—a scenario commonly observed in immunotherapy studies. When properly validated, milestone survival endpoints can serve as primary endpoints. Several meta-analyses \cite{{sargent:2005},{Papamichael:2016}, {vanderVoort:2021}, {halabi:2024}}, have emphasized the value of evaluating time-to-event outcomes at fixed time points, showing that milestone endpoints, such as rPFS, can act as meaningful surrogates for OS and inform trial design, interpretation, and regulatory decision-making.\\

We have also developed an R Shiny application to facilitate practical implementation, allowing users to specify key design parameters -- including hypothesized effect, censoring, power, type I error, interim analyses, and survival distribution (exponential or Weibull) and compute the required sample size efficiently.\\

The primary limitation is the need to pre-specify an appropriate milestone survival. The choice of milestone can influence the estimated treatment effect and required sample size. In our CALGB 90203 trial, selecting a milestone involved careful consideration of dropout and early events. Investigators can mitigate these concerns by using historical data, performing sensitivity analyses across plausible milestones, and ensuring sufficient follow-up to capture outcomes at the selected time.\\

In conclusion, our method provides a flexible framework for milestone-based trial design, overcoming key limitations of conventional approach and providing savings in the required sample size. By incorporating all available follow-up and accommodating non-proportional hazards, it offers investigators a practical, efficient, and interpretable approach for designing modern clinical trials.


\vspace*{-8pt}

\section*{Acknowledgment}

This research was supported in part by the United States Army Medical Research awards HT9425-23-1-0393 and HT9425-25-1-0623, the National Institutes of Health Grants R21 CA263950, R01 LM013352, R01 CA256157, R01 CA249279, the Food and Drug Administration (FDA) Award U01FD007857 of the U.S. Department of Health and Human Services (HHS), and the Prostate Cancer Foundation. All analyses and conclusions in this manuscript are the sole responsibility of the authors and do not necessarily reflect the opinions or views of the NCTN, the NCORP, the NCI or the FDA. 

\section*{Data Availability}
 
 No data were utilized in this manuscript as we conducted simulations. The R shiny app is available \href{https://ynu0gt-chenxi-yu.shinyapps.io/shinyapp/}{Sample Size/Power Calculation of Group Sequential Design in Phase III Clinical Trials}.

\bibliographystyle{natbib}	
\bibliography{ref}

\renewcommand{\thefigure}{S\arabic{figure}}
\setcounter{figure}{0}  
\newpage
\appendix
\renewcommand{\thefigure}{S\arabic{figure}}
\setcounter{figure}{0}  

\section*{Supplementary Materials}
\addcontentsline{toc}{section}{Supplementary Materials}

This supplementary document presents additional simulation results for three combinations of event time distributions and alpha spending approaches. Following Section 3 of the main paper, we evaluate the type I error under the null hypothesis $H_0: S_1(x) = S_2(x) = 0.25$ and the power under the alternative hypothesis $H_A: S_1(x) = 0.25$ and $S_2(x) = 0.25 + D$, where $D = 0.05, 0.1, 0.15, 0.2$ and $x = 1, 5$.

Other than the choice of the event time distribution and the alpha spending approach, all simulation settings in this document are identical to those described in Section 3 of the main paper.

\newpage

\begin{figure}[htp]
\begin{center}
			\scalebox{0.7}{\includegraphics[angle =0]{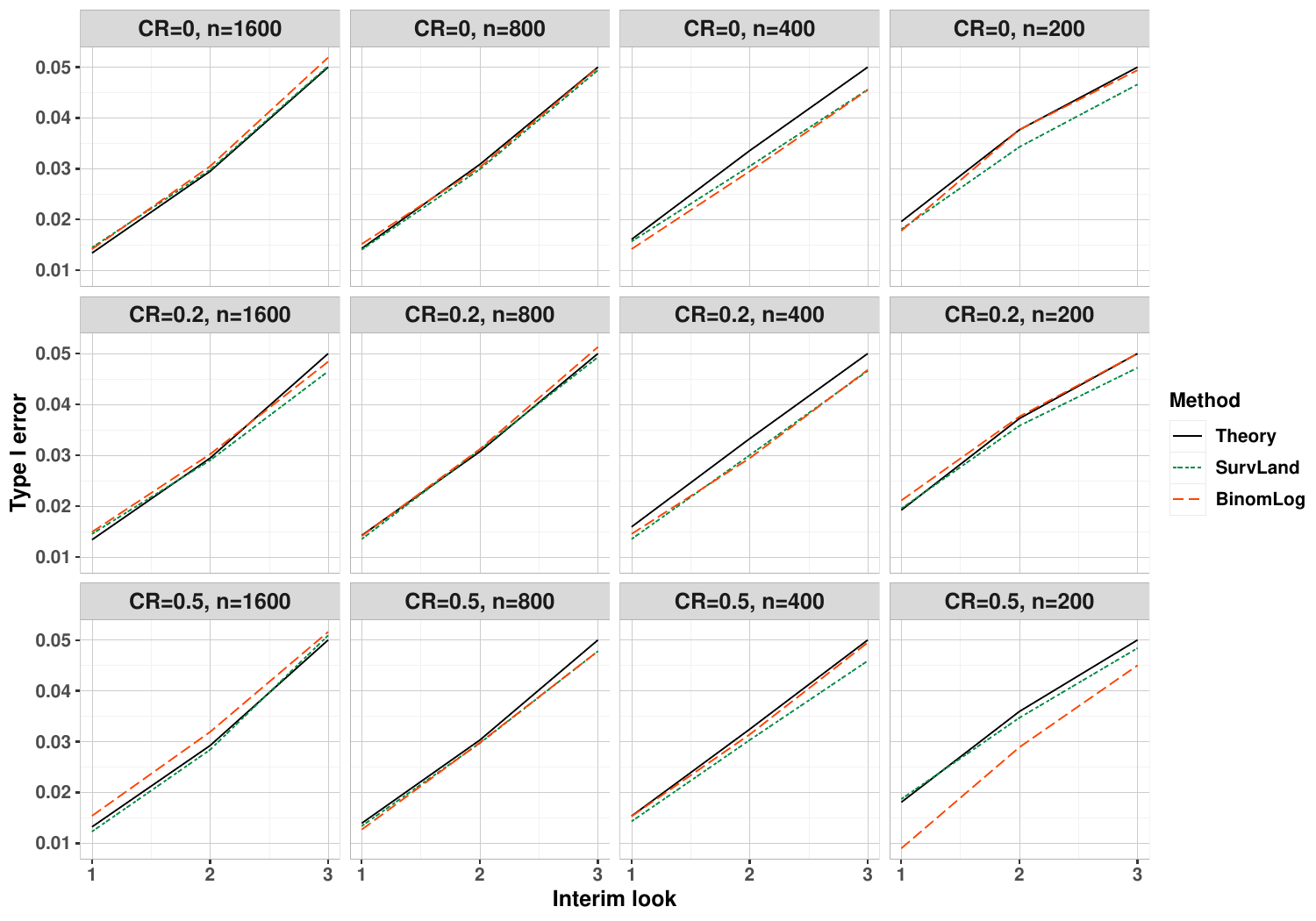}}
\caption{Type I error comparisons between proposed test (SurvLand) and binomial test (BinomLog) based on  survival rate at $x=1$.
}\label{Sfig1}
\end{center}
\end{figure}

\newpage
\begin{figure}[htp]
\begin{center}
			\scalebox{0.7}{\includegraphics[angle =0]{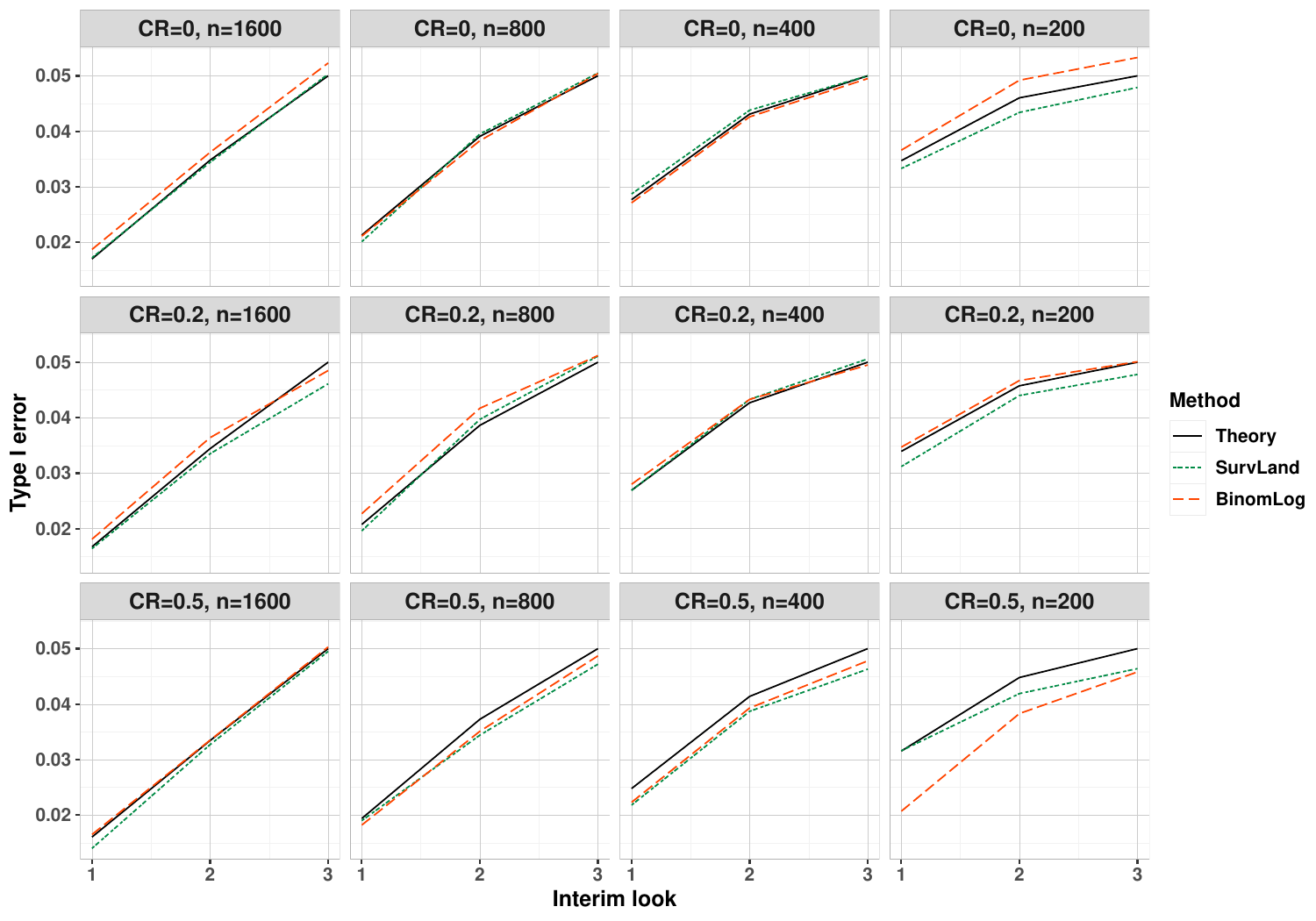}}
\caption{Type I error comparisons between proposed test (SurvLand) and binomial test  (BinomLog) based on survival rate at $x=5$}
\label{Sfig2}
\end{center}
\end{figure}

\newpage
\begin{figure}[htp]
\begin{center}
			\scalebox{0.7}{\includegraphics[angle =0]{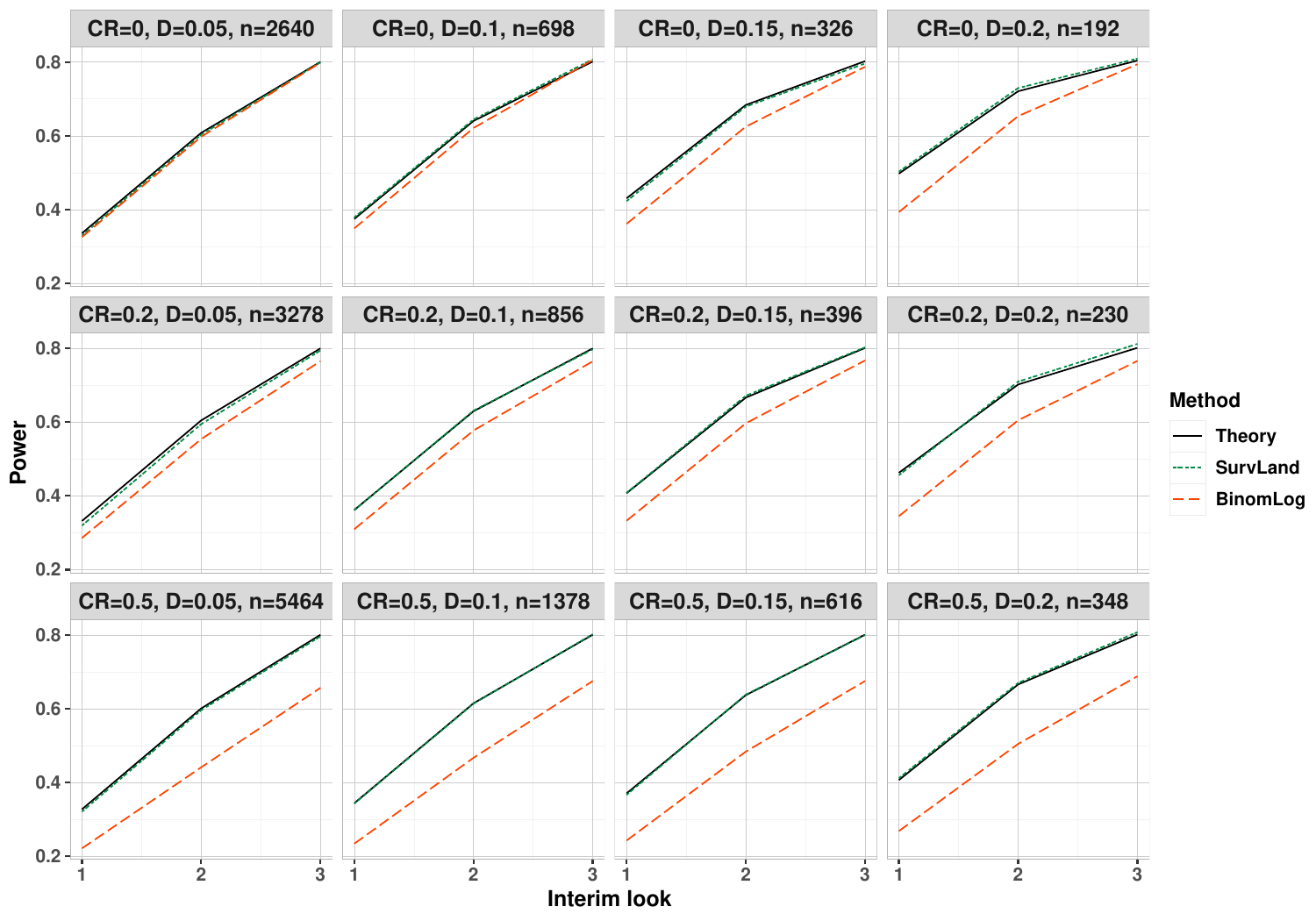}}
\caption{Empirical power comparisons between  proposed test (SurvLand) and  binomial test (BinomLog) with  theoretical power of 80\% based on survival rate at $x=1$ }
\label{Sfig3}
\end{center}
\end{figure}

\newpage
\begin{figure}[htp]
\begin{center}
			\scalebox{0.7}{\includegraphics[angle =0]{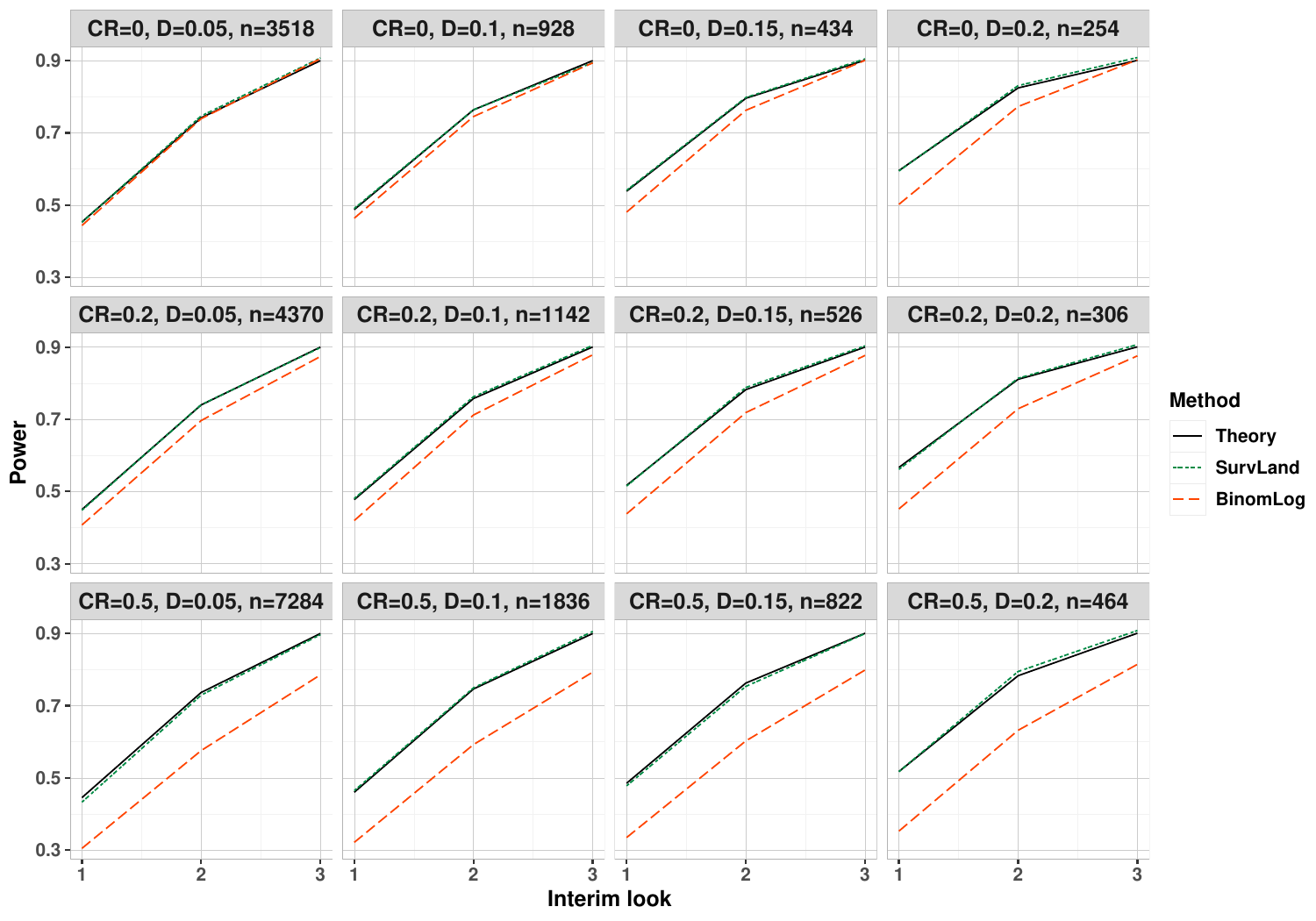}}
\caption{Empirical power comparisons between proposed test (SurvLand) and  binomial test (BinomLog) with  theoretical power of 90\% based on survival rate at $x=1$}
\label{Sfig4}
\end{center}
\end{figure}

\newpage
\begin{figure}[htp]
\begin{center}
			\scalebox{0.7}{\includegraphics[angle =0]{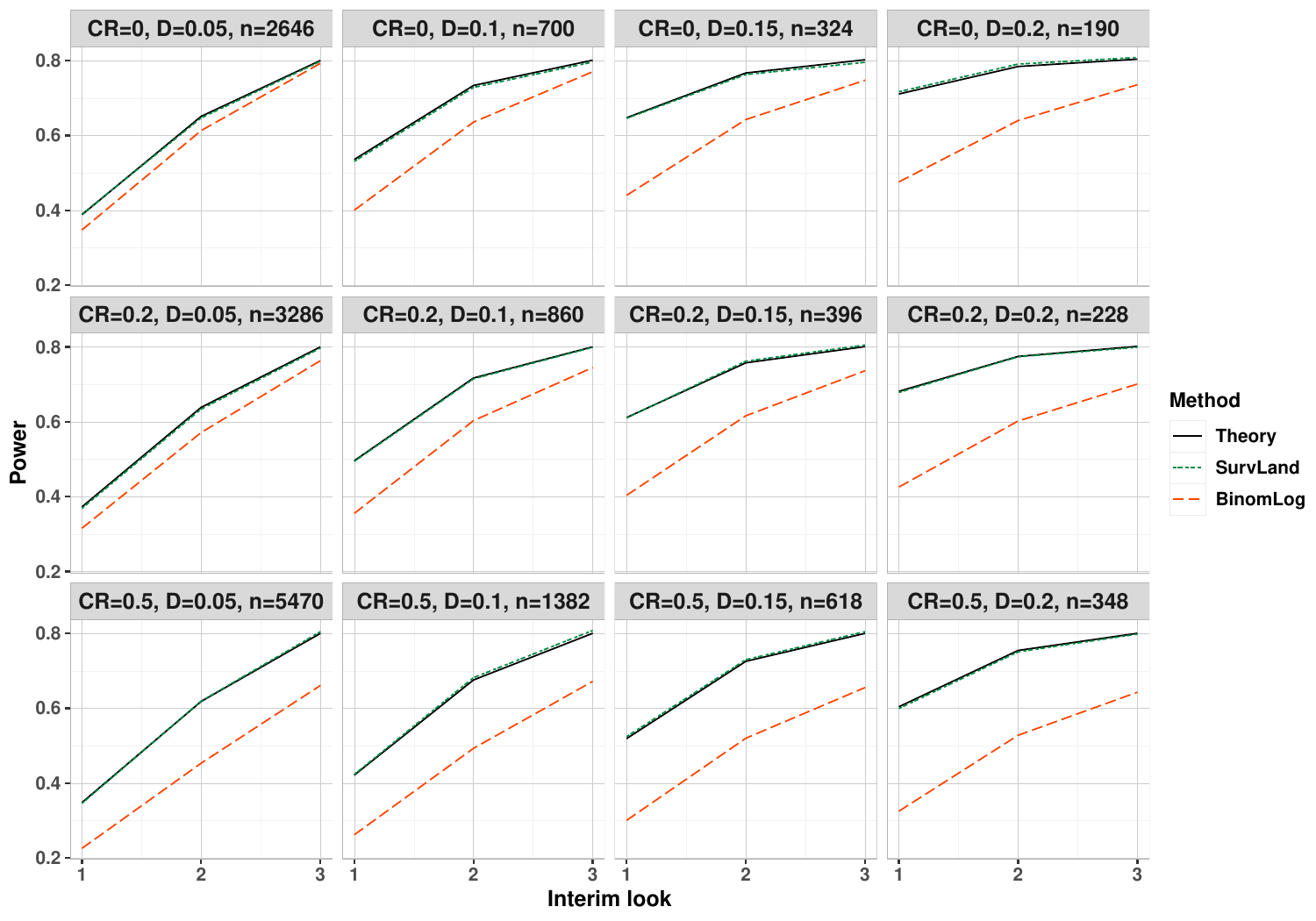}}
\caption{Empirical power comparisons between  proposed test (SurvLand) and  binomial test (BinomLog) with  theoretical power of 80\% based on survival rate at $x=5$ }
\label{Sfig5}
\end{center}
\end{figure}

\newpage
\begin{figure}[htp]
\begin{center}
			\scalebox{0.7}{\includegraphics[angle =0]{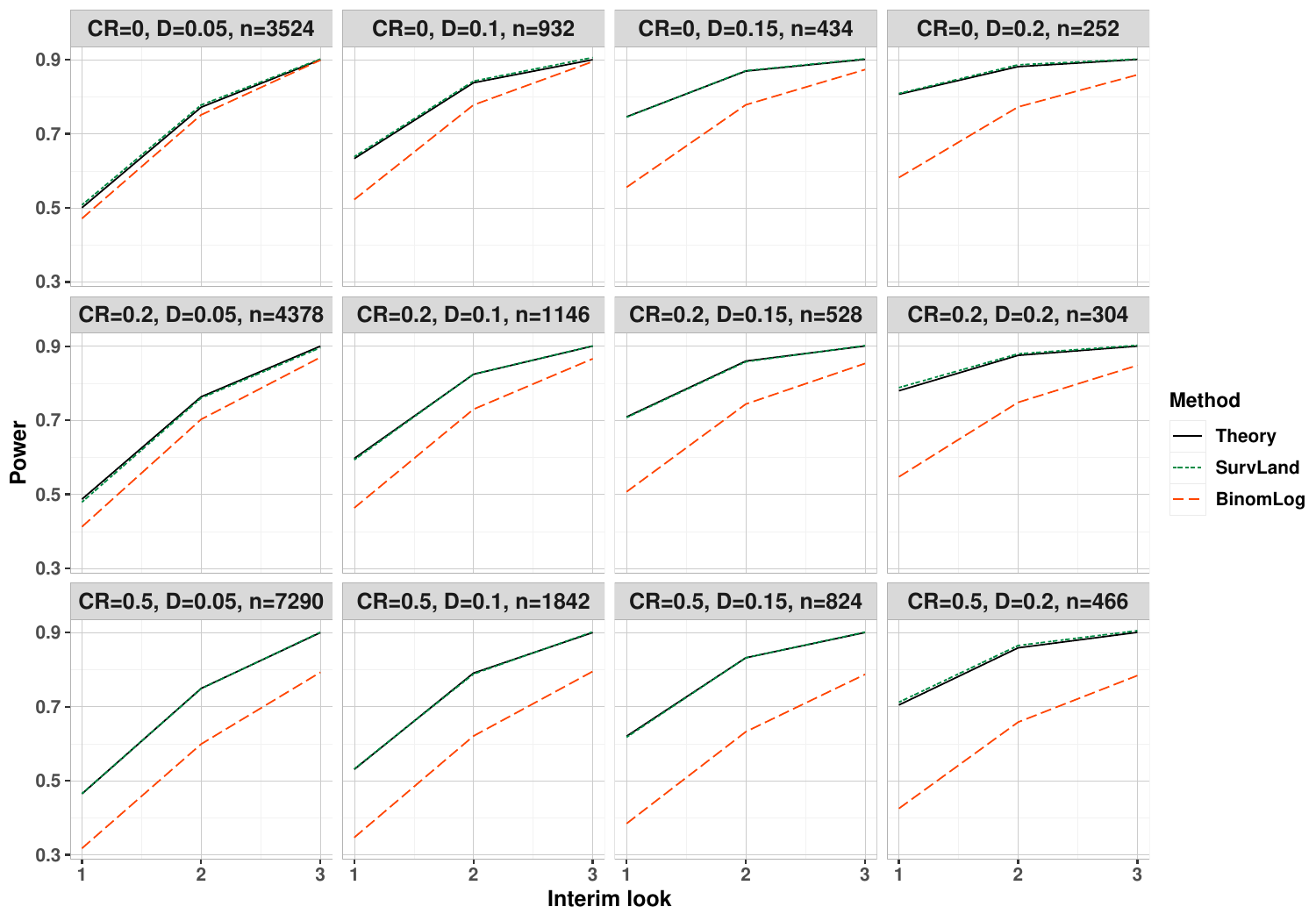}}
\caption{Empirical power comparisons between proposed test (SurvLand) and  binomial test (BinomLog) with  theoretical power of 90\% based on survival rate at $x=5$}
\label{Sfig6}
\end{center}
\end{figure}

\newpage
\section{Results based on exponential distribution with O'Brien-Fleming approach}
\begin{figure}[htp]
\begin{center}
			\scalebox{0.7}{\includegraphics[angle =0]{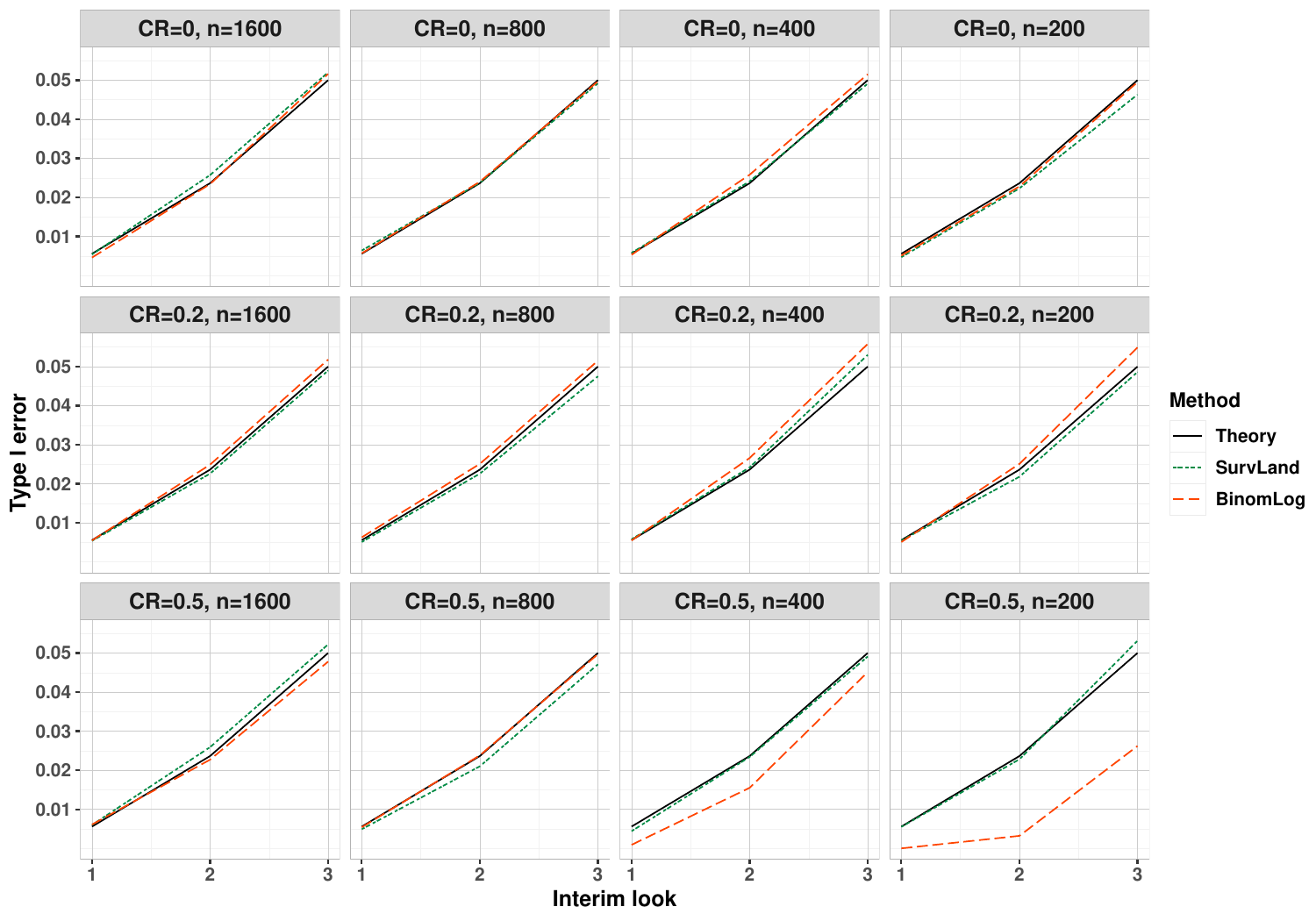}}
\caption{Type I error comparisons between proposed test (SurvLand) and binomial test (BinomLog) based on  survival rate at $x=1$.}
\label{Sfig7}
\end{center}
\end{figure}

\newpage
\begin{figure}[htp]
\begin{center}
			\scalebox{0.7}{\includegraphics[angle =0]{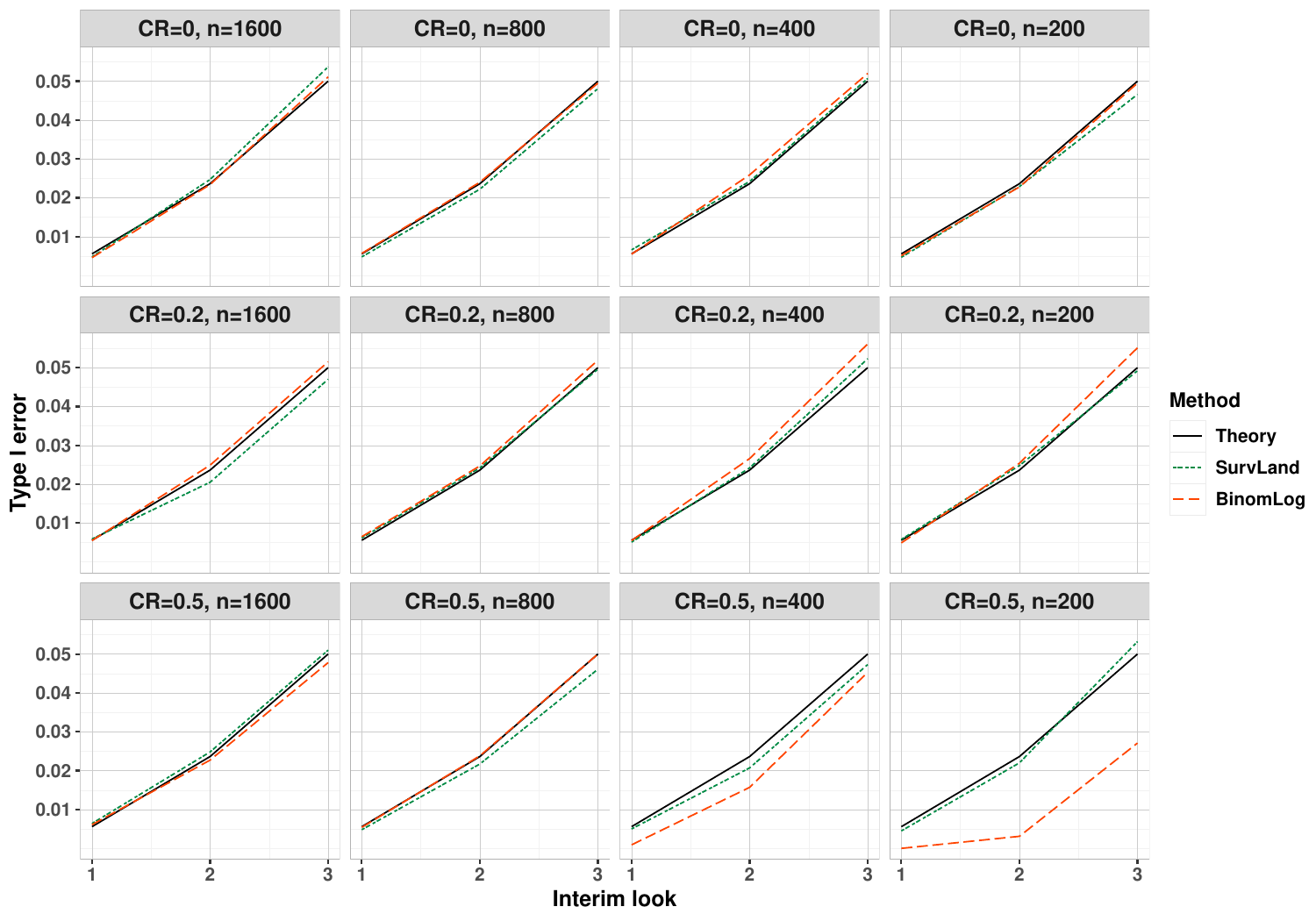}}
\caption{Type I error comparisons between proposed test (SurvLand) and binomial test  (BinomLog) based on survival rate at $x=5$}
\label{Sfig8}
\end{center}
\end{figure}

\newpage
\begin{figure}[htp]
\begin{center}
			\scalebox{0.7}{\includegraphics[angle =0]{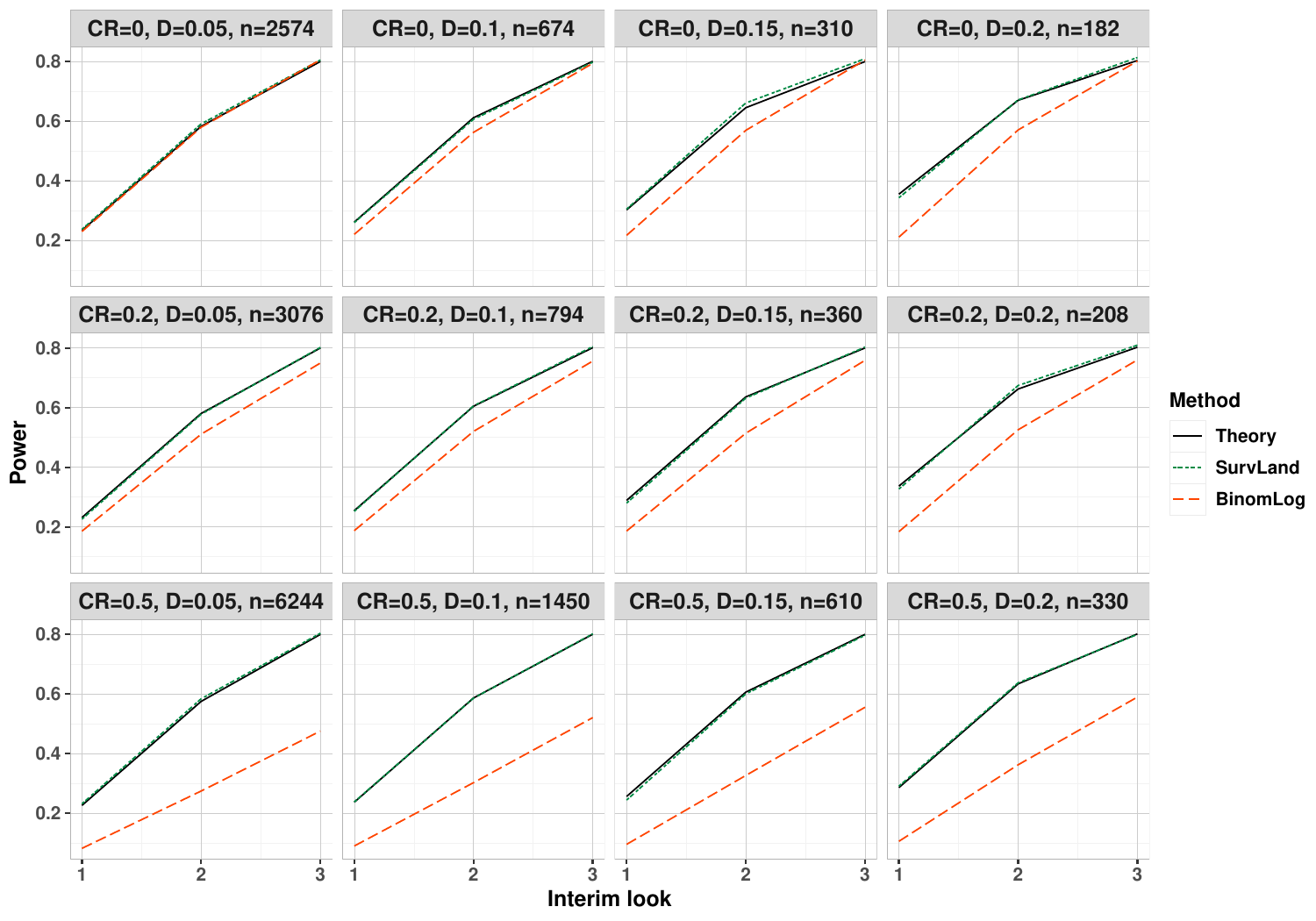}}
\caption{Empirical power comparisons between  proposed test (SurvLand) and  binomial test (BinomLog) with  theoretical power of 80\% based on survival rate at $x=1$ }
\label{Sfig9}
\end{center}
\end{figure}

\newpage
\begin{figure}[htp]
\begin{center}
			\scalebox{0.7}{\includegraphics[angle =0]{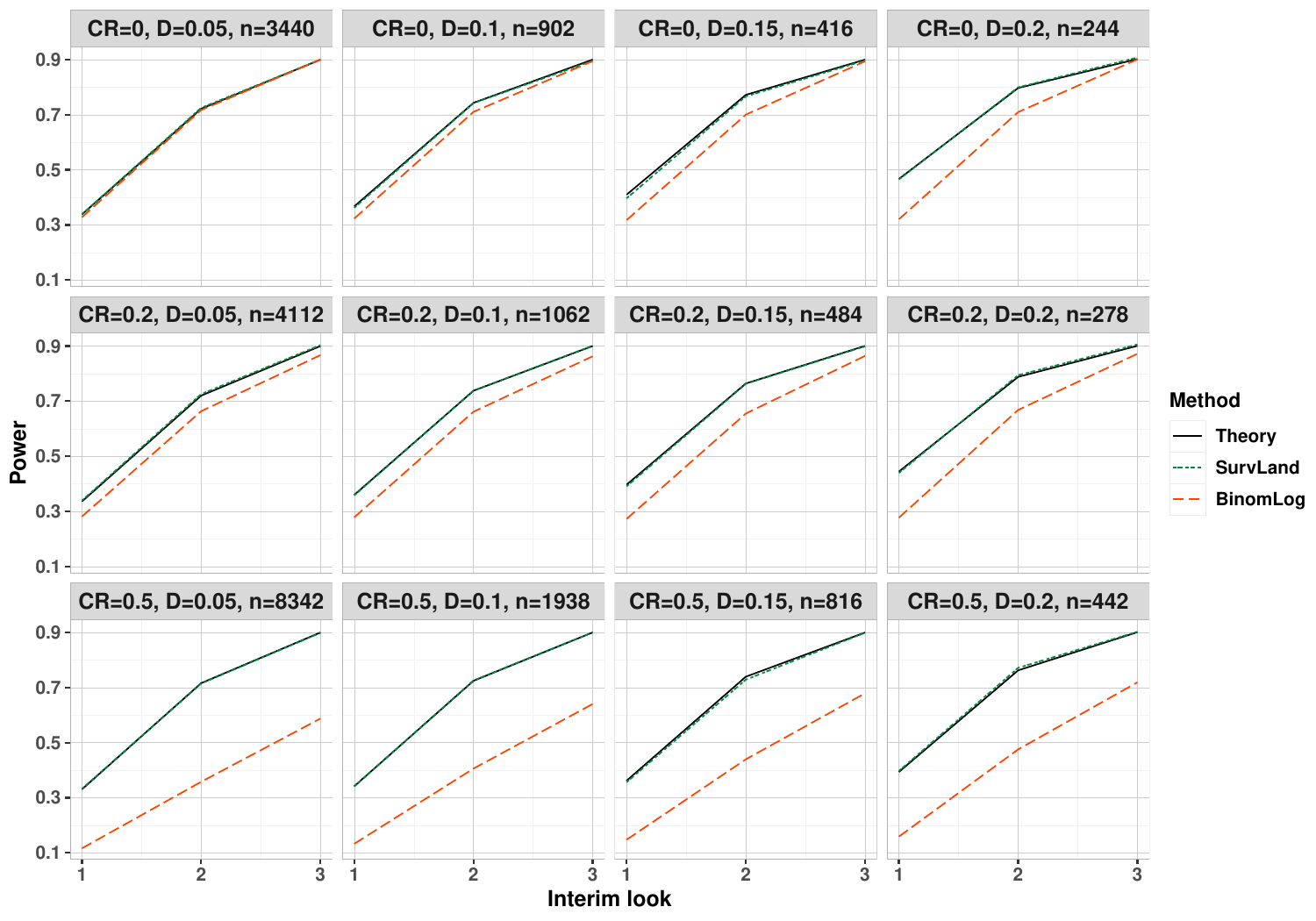}}
\caption{Empirical power comparisons between proposed test (SurvLand) and  binomial test (BinomLog) with  theoretical power of 90\% based on survival rate at $x=1$}
\label{Sfig10}
\end{center}
\end{figure}

\newpage
\begin{figure}[htp]
\begin{center}
			\scalebox{0.7}{\includegraphics[angle =0]{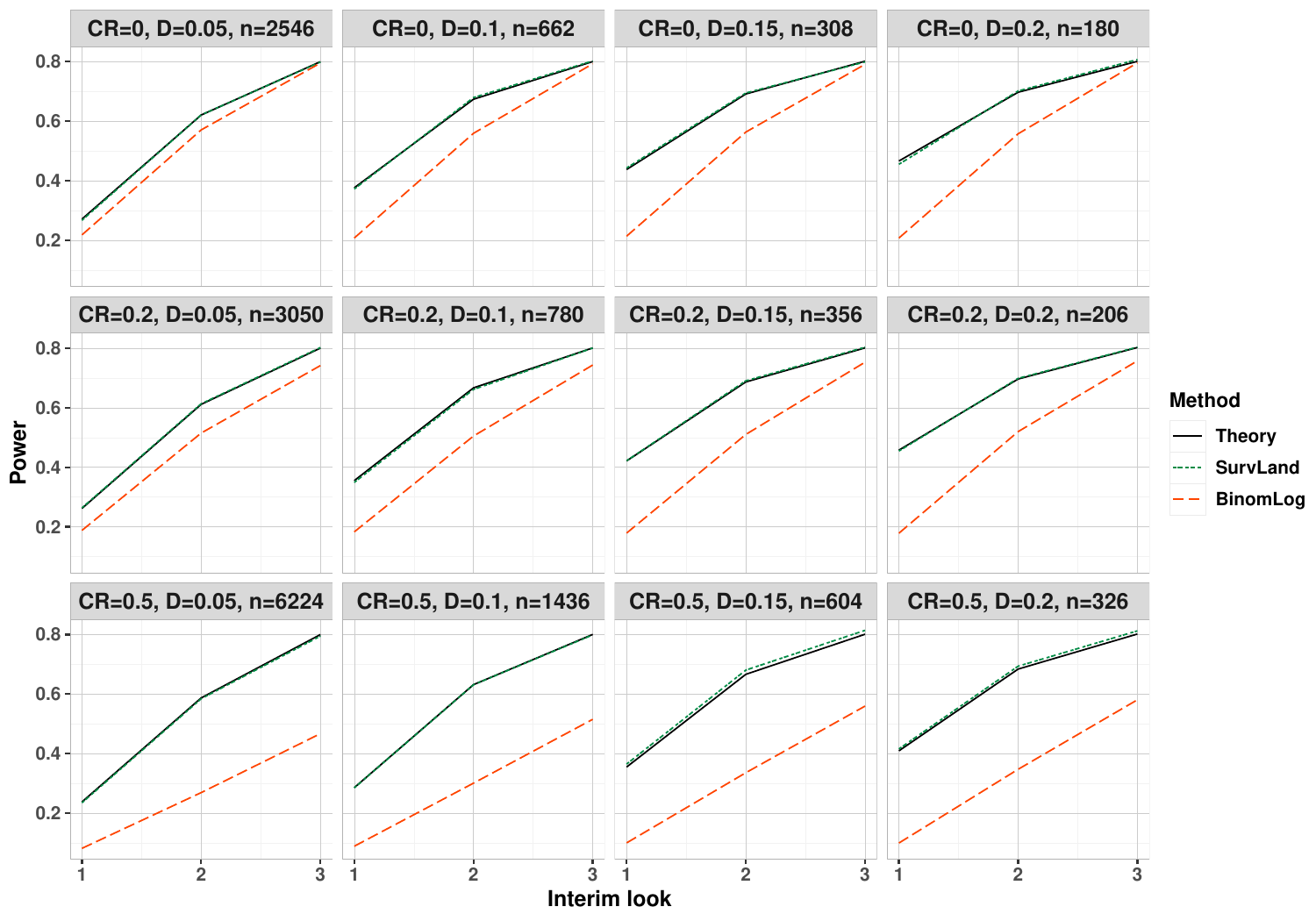}}
\caption{Empirical power comparisons between  proposed test (SurvLand) and  binomial test (BinomLog) with  theoretical power of 80\% based on survival rate at $x=5$ }
\label{Sfig11}
\end{center}
\end{figure}

\newpage
\begin{figure}[htp]
\begin{center}
			\scalebox{0.7}{\includegraphics[angle =0]{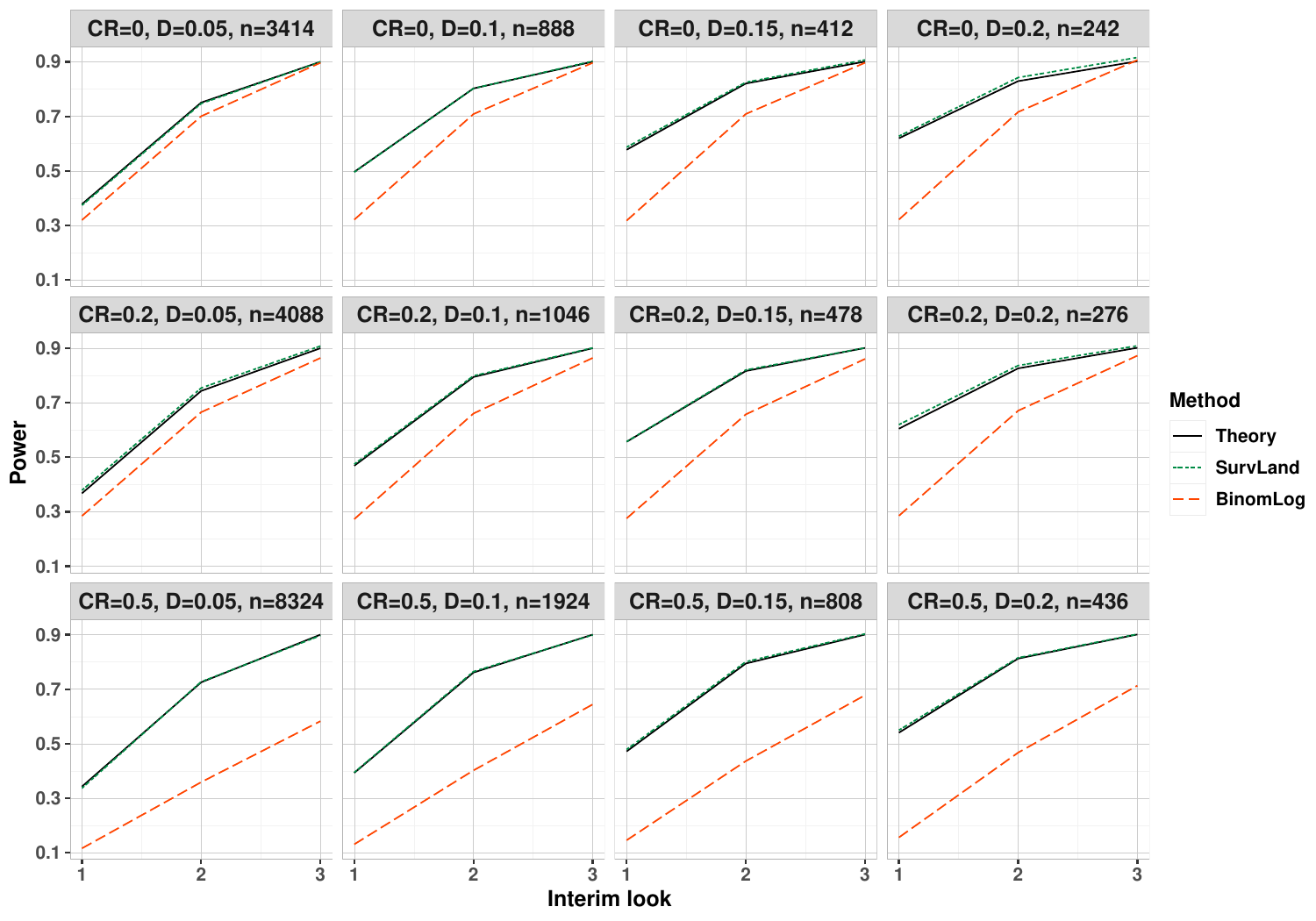}}
\caption{Empirical power comparisons between proposed test (SurvLand) and  binomial test (BinomLog) with  theoretical power of 90\% based on survival rate at $x=5$}
\label{Sfig12}
\end{center}
\end{figure}

\newpage
\section{Results based on exponential distribution with Lan-DeMets approach}
\begin{figure}[htp]
\begin{center}
			\scalebox{0.7}{\includegraphics[angle =0]{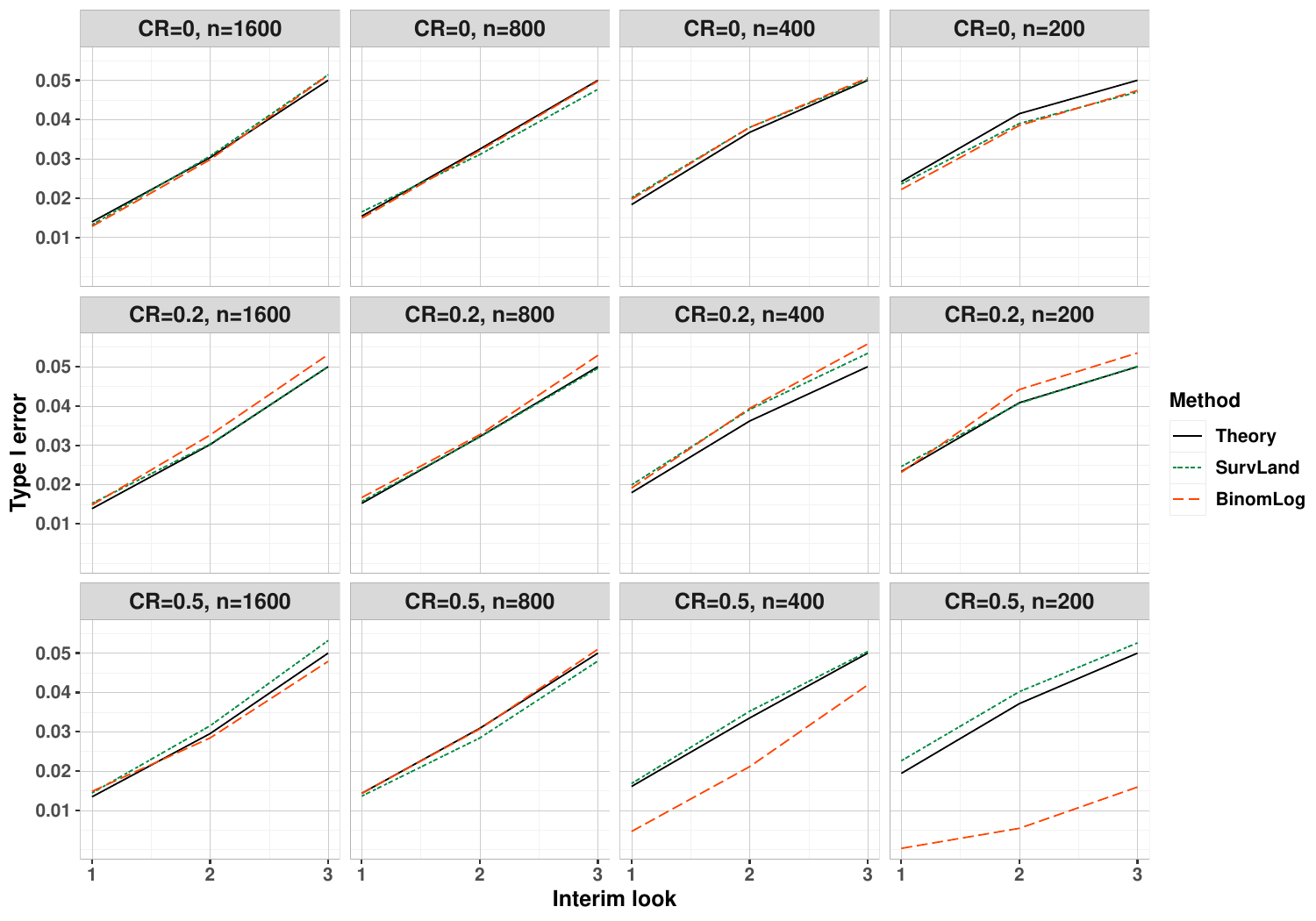}}
\caption{Type I error comparisons between proposed test (SurvLand) and binomial test (BinomLog) based on  survival rate at $x=1$.}
\label{Sfig13}
\end{center}
\end{figure}

\newpage
\begin{figure}[htp]
\begin{center}
			\scalebox{0.7}{\includegraphics[angle =0]{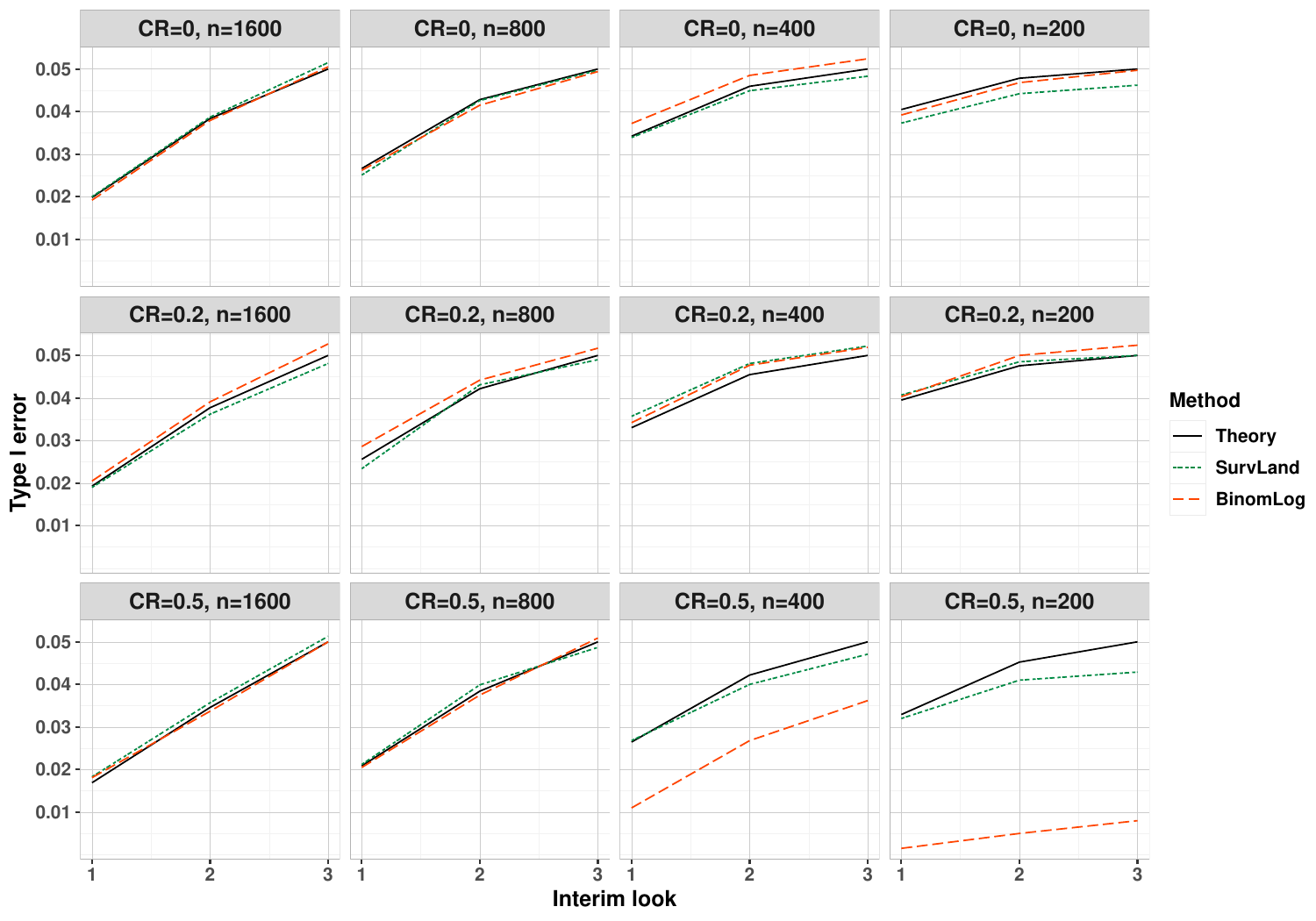}}
\caption{Type I error comparisons between proposed test (SurvLand) and binomial test  (BinomLog) based on survival rate at $x=5$}
\label{Sfig14}
\end{center}
\end{figure}

\newpage
\begin{figure}[htp]
\begin{center}
			\scalebox{0.7}{\includegraphics[angle =0]{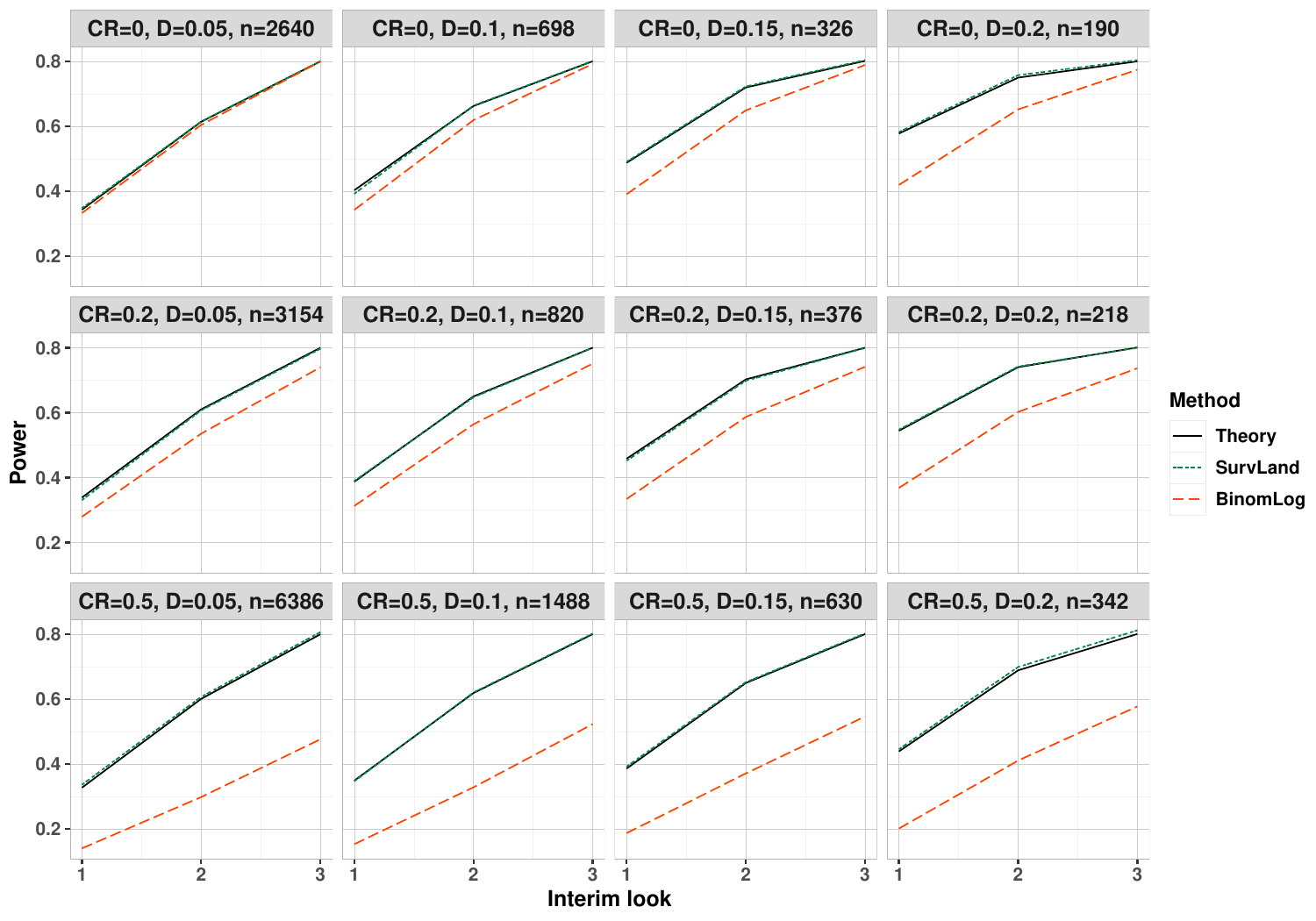}}
\caption{Empirical power comparisons between  proposed test (SurvLand) and  binomial test (BinomLog) with  theoretical power of 80\% based on survival rate at $x=1$}
\label{Sfig15}
\end{center}
\end{figure}

\newpage
\begin{figure}[htp]
\begin{center}
			\scalebox{0.7}{\includegraphics[angle =0]{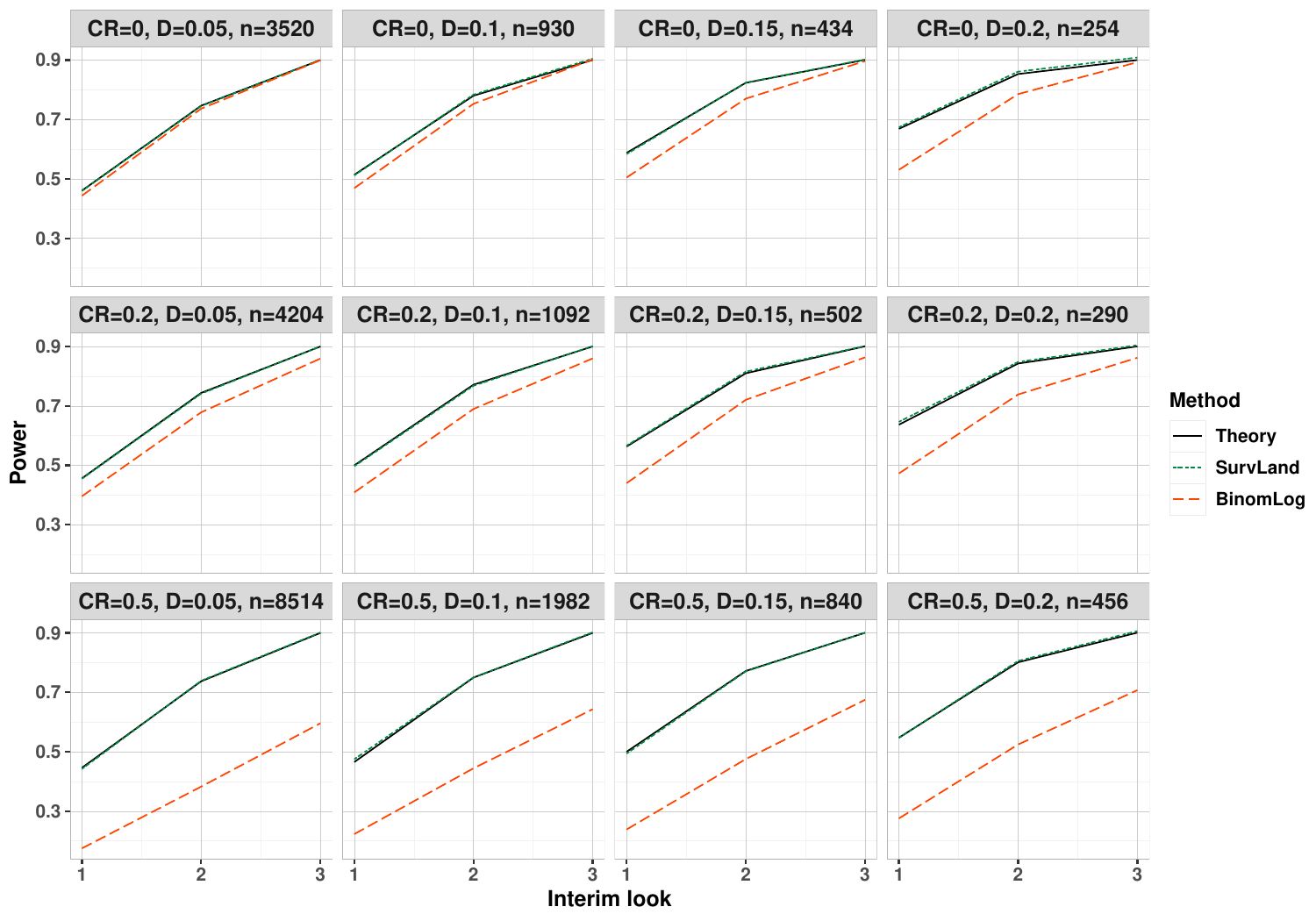}}
\caption{Empirical power comparisons between proposed test (SurvLand) and  binomial test (BinomLog) with  theoretical power of 90\% based on survival rate at $x=1$}
\label{Sfig16}
\end{center}
\end{figure}

\newpage
\begin{figure}[htp]
\begin{center}
			\scalebox{0.7}{\includegraphics[angle =0]{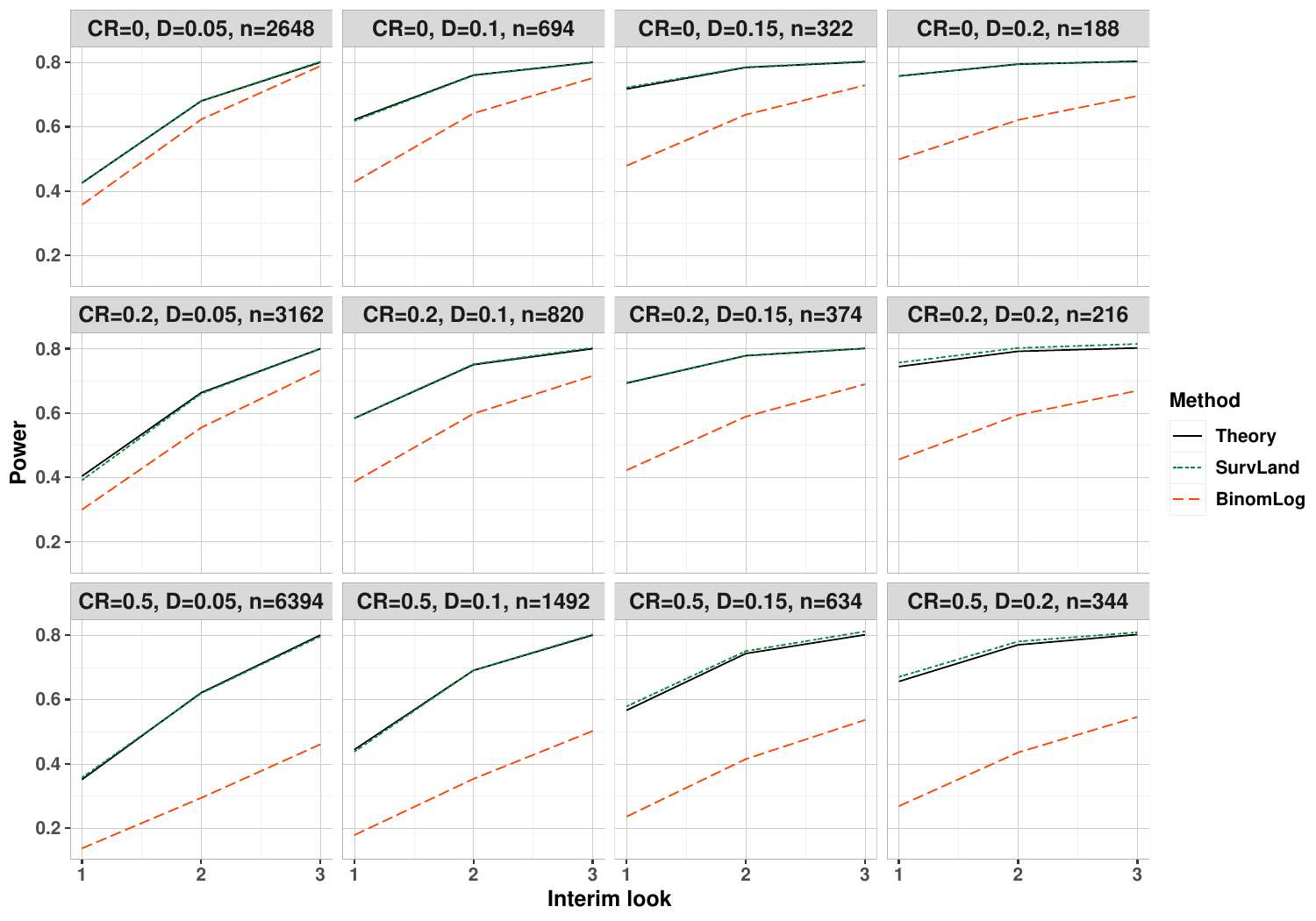}}
\caption{Empirical power comparisons between  proposed test (SurvLand) and  binomial test (BinomLog) with  theoretical power of 80\% based on survival rate at $x=5$ }
\label{Sfig17}
\end{center}
\end{figure}

\newpage
\begin{figure}[htp]
\begin{center}
			\scalebox{0.7}{\includegraphics[angle =0]{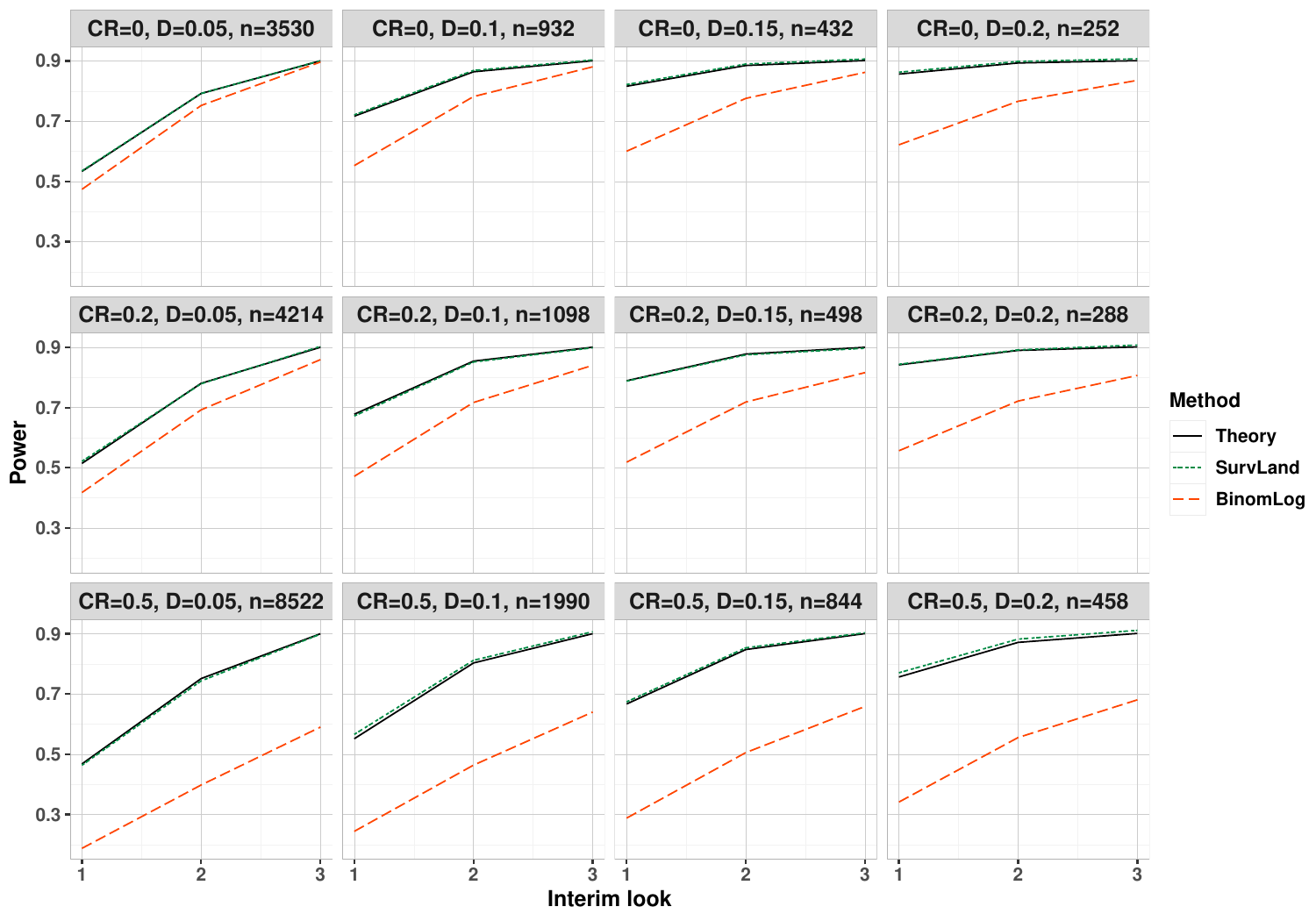}}
\caption{Empirical power comparisons between proposed test (SurvLand) and  binomial test (BinomLog) with  theoretical power of 90\% based on survival rate at $x=5$}
\label{Sfig18}
\end{center}
\end{figure}

\end{document}